\documentclass[11pt,draft]{article}

\usepackage[T1]{fontenc}

\usepackage{amsmath,amssymb,amsthm}

\newcommand{\R}{\mathbb{R}}

\newcommand{\LL}{{\cal L}}
\newcommand{\D}{{\cal D}}

\newcommand{\h}{{\cal H}}

\newcommand{\eps}{\epsilon}
\newcommand{\veps}{\varepsilon}
\newcommand{\ot}{\otimes}

\newcommand{\bld}[1]{\boldsymbol{#1}}
\newcommand{\tl}[1]{\tilde{#1}}

\DeclareMathOperator{\sgn}{{\rm sgn}}
\newcommand{\we}{\wedge}

\newcommand{\bth}{\boldsymbol{\theta}}

\newcommand{\lr}{\lrcorner\,}

\newcommand{\bed}{\boldsymbol{d}}

\newcounter{mnotecount}[section]

\newtheorem{thr}{Theorem}
\newtheorem{lm}[thr]{Lemma}

\numberwithin{equation}{section}
\numberwithin{thr}{section}

\begin{document}

\title{Variables suitable for constructing quantum states for the Teleparallel Equivalent of General Relativity II\footnote{This is an author-created version of a paper published as {\em Gen. Rel. Grav.} {\bf 46} 1638 (2014) DOI 10.1007/s10714-013-1638-2}}
\author{ Andrzej Oko{\l}\'ow}
\date{July 14, 2014}

\maketitle
\begin{center}
{\it  Institute of Theoretical Physics, Warsaw University\\ ul. Pasteura 5, 02-093 Warsaw, Poland\smallskip\\
oko@fuw.edu.pl}
\end{center}
\medskip

\begin{abstract}
We present the second (and final) part of an analysis aimed at introducing variables which are suitable for constructing a space of quantum states for the Teleparallel Equivalent of General Relativity. In the first part of the analysis we introduced  a family of variables on the ``position'' sector of the phase space. In this paper we distinguish differentiable variables in the family. Then we define momenta conjugate to the differentiable variables and express constraints of the theory in terms of the variables and the momenta. Finally, we exclude variables which generate an obstacle for further steps of the Dirac's procedure of canonical quantization of constrained systems we are going to apply to the theory. As a result we obtain two collections of variables on the phase space which will be used (in a subsequent paper) to construct the desired space of quantum states. 
\end{abstract}

\section{Introduction}

In \cite{q-suit} we were searching for variables on the phase space of  the Teleparallel Equivalent of General Relativity (TEGR) which are suitable for constructing a space of kinematic quantum states for the theory via projective methods described in \cite{q-nonl}---the space of quantum states is meant to be used in a quantization of TEGR according to the Dirac's approach to canonical quantization of constrained systems. Let us briefly recall the results of \cite{q-suit}.  

The phase space of TEGR can be seen \cite{oko-tegr} as a Cartesian product $P\times \Theta$ of a space $P$ of momenta and a space $\Theta$ of configuration (``position'') variables, where $P$ consists of quadruplets of two-forms $(p_A)$, $A=0,1,2,3$, on a three-dimensional manifold $\Sigma$, while $\Theta$ does of quadruplets of one-forms $(\theta^B)$, $B=0,1,2,3$, subjected to a restriction. We will call  the fields $(p_A,\theta^B)$ {\em natural variables} on the phase space. 

It was shown in \cite{q-suit} that it is possible to construct via the projective methods a space of quantum states for TEGR using the natural variables. However, the resulting space turned out to be too large in the sense that it does not correspond strictly enough to the phase space but it does to a space essentially larger than the phase space. Therefore we were forced to look for some other variables on the phase space.

As a result we found in \cite{q-suit} a family $\{(\xi^I_\iota,\theta^J)\}$ of variables on the configuration space $\Theta$. Here $\iota$ is a parameter distinguishing elements of the family---$\iota$ is a special function on $\Theta$ valued in a set $\{-1,1\}$. Given $\iota$, $(\xi^I_\iota)$, $I=1,2,3$, is a triplet of real functions on $\Sigma$, while $(\theta^J)$, $J=1,2,3$, is a triplet of one-forms being a global coframe on the manifold. We will call the variables $(\xi^I_\iota,\theta^J)$ (and momenta conjugate to them) {\em new variables} on $\Theta$ (on the phase space). We showed that for every $\iota$ the variables $(\xi^I_\iota,\theta^J)$ satisfy some conditions which indicate that perhaps the variables may be used to construct a space of quantum states for TEGR via the projective methods. 

Now let us describe the goals of the present paper.

Introducing in \cite{q-suit} the family of new variables we neglected an issue of differentiability of the variables with respect to the natural ones---therefore it may happen that functions on the phase space like constraints and a Hamiltonian being differentiable functions (in the sense of variational calculus) of the natural variables are not differentiable functions of the new ones. 

Thus one goal of this paper is to find a criterion which will allow us to recognize differentiable variables in the family $\{(\xi^I_\iota,\theta^J)\}$. 

Moreover, we would like to find variables $(\xi^I_\iota,\theta^J)$ which not only provide a space of kinematic quantum states for TEGR but which provide a {\em useful} space of such states. To explain what we mean by ``useful'' let us recall that TEGR is a constrained system (see e.g. \cite{bl,maluf-1,maluf,oko-tegr}) and since the constraints on the phase space of TEGR are too complicated to be solved classically we are going to apply the Dirac's approach to quantize TEGR. According to the Dirac's approach one first constructs a space of quantum states corresponding to the unconstrained phase space, that is, a space of {\em kinematic} quantum states. Then among the kinematic quantum states one distinguishes {\em physical} quantum states as ones corresponding to these classical states which satisfy all constraints---in other words, one imposes ``quantum constraints'' on the kinematic quantum states. If, given space of kinematic quantum states for TEGR, it is possible to define a (workable) procedure which isolates physical quantum states from the kinematic ones then we consider this space to be {\em useful}.

The problem of defining such a workable procedure will be not solved it in this paper. Nevertheless, we will show that some new variables are problematic in the following sense: even if these variables provide a space of kinematic quantum states for TEGR then there appears an obstacle for imposing quantum constraints on this space. To describe the obstacle let us recall that according to the general construction \cite{q-nonl} the space of quantum states would be built from some functions on the phase space called {\em elementary degrees of freedom} and in \cite{q-suit} we chose d.o.f. as functions naturally defined by new variables $(\xi^I_\iota,\theta^J)$. We will show that in cases of some new variables some constraints of TEGR cannot be even approximated by any finite number of the d.o.f.. Consequently, quantum constraints cannot be imposed on any sector of the space of quantum states given by a finite number of corresponding quantum d.o.f.. This means that the quantum constraints would have to be imposed directly on the whole space or on its sectors each given by an infinite number of quantum d.o.f.. We will argue that this fact makes the task of isolating physical quantum states from the space very hard (if not impossible). But we will show also that there exist precisely two functions $\iota_1,\iota_2:\Theta\to\{-1,1\}$ such that the problem just described does not appear in the case of variables $(\xi^I_{\iota_1},\theta^J)$ and $(\xi^I_{\iota_2},\theta^J)$. 

In \cite{q-stat} we will use these variables to construct a space of kinematic quantum states for TEGR\footnote{Actually, it will turn out that $\iota_1=-\iota_2$. A consequence of this fact is that  $\xi^I_{\iota_1}=-\xi^I_{\iota_2}$ and the space of quantum states constructed from $(\xi^I_{\iota_1},\theta^J)$ coincides with that constructed from $(\xi^I_{\iota_2},\theta^J)$.} which, hopefully, will turn out to be useful.           

Thus the other goals of the paper are, given function $\iota$, 
\begin{enumerate}
\item to express the momenta $(p_A)$ as functions of the new variables on the phase space,
\item to rewrite the constraints in terms of the new variables,
\item to check whether there are obstacles for approximating the constraints by means of  finite numbers of d.o.f. defined by $(\xi^I_{\iota},\theta^J)$.
\end{enumerate}

Obviously, the constraints of TEGR should be expressed in terms of new variables not only in order to exclude problematic variables---this is also a preparatory step for defining quantum constraints on the space of quantum states we are going to construct. 

Let us emphasize that the tasks 2 and 3 above will be also completed for so-called Yang-Mills-type Teleparallel Model (YMTM). This is a theory of the same phase space as TEGR but of simpler dynamics \cite{itin,os} which may be useful as a toy-model for testing some elements of quantization procedure before they will be applied to TEGR. 

The paper is organized as follows: in Section 2 we introduce basic definitions, present a precise description of the phase space of TEGR and a definition of new variables $(\xi^I_\iota,\theta^J)$, derive some auxiliary formulae which will be used in further parts of the paper, finally we express the new variables in terms of the natural ones. In short Section 3 we address the issue of differentiability of $(\xi^I_\iota,\theta^J)$. In Section 4 we derive formulae describing the momenta $(p_A)$ as functions of any (differentiable) new variables on the phase space and formulae describing the momenta conjugate to  $(\xi^I_{\iota},\theta^J)$ as functions of the natural variables. Moreover, in this section we present the constraints (and Hamiltonians) of TEGR and of YMTM expressed in terms of new variables and check in which cases there appears the obstacle mentioned above. Section 5 contains a summary and a short discussion of results obtained in this paper. Finally, in Appendix \ref{app-TH-bij} we prove a useful lemma and in Appendix \ref{constr-der} we derive formulae expressing the constraints of TEGR and YMTM in terms of new variables on the phase space.

\section{Preliminaries}

\subsection{Vector spaces and scalar products \label{vs-sp}}

Let $\mathbb{M}$ be a four-dimensional oriented vector space equipped with a scalar product $\eta$ of signature $(-,+,+,+)$. We fix an orthonormal basis $(v_A)$ $(A=0,1,2,3)$ such that the components $(\eta_{AB})$ of $\eta$ given by the basis form the matrix ${\rm diag}(-1,1,1,1)$. The matrix $(\eta_{AB})$ and its inverse $(\eta^{AB})$ will be used to, respectively, lower and raise capital Latin letter indeces $A,B,C,D\in\{0,1,2,3\}$. The scalar product $\eta$ defines a volume form on $\mathbb{M}$---components of the form in the basis $(v_A)$ will be denoted by $\veps_{ABCD}$.

Let $\mathbb{E}$ be a subspace of $\mathbb{M}$ spanned by the vectors $\{v_1,v_2,v_3\}$. 
The scalar product $\eta$ induces on $\mathbb{E}$ a positive definite scalar product $\delta$---its components $(\delta_{IJ})$ in the basis $(v_1,v_2,v_3)$ form the matrix ${\rm diag}(1,1,1)$. The matrix $(\delta_{IJ})$ and its inverse $(\delta^{IJ})$ will be used to, respectively, lower and raise capital Latin letter indeces $I,J,K,L,M\in\{1,2,3\}$. 

We will denote by $\veps_{IJK}$ the three-dimensional permutation symbol. Note that 
\[
\veps_{0\,IJK}=\veps_{IJK}.
\]

\subsection{Phase space \label{ph-sp}}

Let $\Sigma$ be a three-dimensional compact manifold without boundary. Throughout the paper this manifold will represent a space-like slice of a spacetime. 

The phase space of TEGR is a Cartesian product $P\times\Theta$, where \cite{oko-tegr}
\begin{enumerate}
\item $\Theta$ consists of all quadruplets of one-forms $(\theta^A)$ on $\Sigma$ such that the metric
\begin{equation}
q=\eta_{AB}\theta^A\ot\theta^B
\label{q}
\end{equation}
is {\em Riemannian} (positive definite);
\item $P$ is a space of all quadruplets of two-forms $(p_A)$ on $\Sigma$.
\end{enumerate}     
The two-form $p_A$ play the role of momenta conjugate to $\theta^A$. The space $\Theta$ will be called {\em {Hamiltonian} configuration space}, $P$---{\em momentum space}. Recall that in the introduction $(p_A,\theta^B)$ were called {\em natural variables} on the phase space. The Poisson bracket of $F$ and $G$ being functions on the phase space reads
\[
\{F,G\}=\int_\Sigma\Big(\frac{\delta F}{\delta {\theta}^A}\we\frac{\delta G}{\delta p_A}-\frac{\delta G}{\delta {\theta}^A}\we\frac{\delta F}{\delta p_A}\Big).
\]  

\subsection{New variables on the Hamiltonian configuration space \label{nv-Th}}

In \cite{q-suit} we introduced {\em new variables} $(\xi^I_\iota,\theta^J)$  on $\Theta$:
\begin{lm}
Given function $\iota$ defined on the space of all global coframes on $\Sigma$ and valued in the set $\{-1,1\}$, there exists a one-to-one correspondence between elements of $\Theta$ and all pairs $(\xi^I_\iota,\theta^J)$ consisting of
\begin{enumerate}
\item functions $\xi_\iota^I$ ($I=1,2,3$)  on $\Sigma$, 
\item one-forms  $\theta^J$ ($J=1,2,3$) on $\Sigma$  constituting a global coframe on the manifold. 
\end{enumerate}  
The correspondence is given by 
\begin{equation}
(\xi^I_\iota,\theta^J)\mapsto \Big(\theta^0=\iota(\theta^L)\frac{\xi_{\iota I}}{\sqrt{1+\xi_{\iota K}\xi_{\iota}^K}}\theta^I,\theta^J\Big)\in\Theta.
\label{th-bij-eq}
\end{equation}
\label{th-bij}
\end{lm}

It follows from the lemma that if $(\theta^A)=(\theta^0,\theta^J)\in\Theta$ then $(\theta^J)$ is a {\em global coframe} on $\Sigma$. There are some important consequences of this fact:
\begin{enumerate}
\item given a triplet $(\theta^J)$ coming from $(\theta^A)\in\Theta$ we can associate with it a number 
\[
\sgn(\theta^J):=
\begin{cases}
1 & \text{if $(\theta^J)$ is compatible with the orientation of $\Sigma$}\\
-1 & \text{otherwise}
\end{cases}.
\]
\item each function $\iota$ can be treated as a function defined on $\Theta$ and valued in $\{-1,1\}$.
\item $\Theta$ splits into two disjoint sets $\Theta_+$ and $\Theta_-$, where $\Theta_+$ ($\Theta_-$) is  constituted by all quadruplets $(\theta^0,\theta^J)$ such that $(\theta^J)$ is compatible (incompatible) with the orientation of $\Sigma$. 
\end{enumerate}

Let us recall a useful interpretation of the variables $(\xi^I_\iota)$ \cite{q-suit}. Given $(\theta^A)$, consider the following equations \cite{nester}
\begin{align}
&\xi_A\theta^A=0, &&\xi^A\xi_A=-1
\label{xi-df}
\end{align}
imposed on a function $(\xi^A)$ on $\Sigma$ valued in $\mathbb{M}$. Note that there exists exactly two continuous solutions of these equations---indeed, the values of the function are normed time-like vectors in $\mathbb{M}$ and therefore the value of $\xi^0$ must be non-zero everywhere. Continuity of $(\xi^A)$ means that the time-like component $\xi^0$ of $(\xi^A)$ is either a positive or a negative function. Consequently, the two continuous solutions of \eqref{xi-df} can be distinguished by $\sgn(\xi^0)$ being the sign of $\xi^0$. 

Suppose that new variables $(\xi^I_\iota,\theta^J)$ corresponds to $(\theta^A)\in\Theta$ according to \eqref{th-bij}. Then $(\xi^I_\iota)$ are equal to the space-like components of this solution $(\xi^A)$ for which $\sgn(\xi^0)=\iota(\theta^J)$. More precisely, 
\begin{equation}
(\xi_\iota^A)=(\xi_\iota^0,\xi_\iota^I),
\label{xiA-xiImu}
\end{equation}
where
\begin{equation}
\xi_\iota^0=\iota(\theta^J)\sqrt{1+\xi_{\iota K}\xi_\iota^K} 
\label{xi-0}
\end{equation}
is a solution of \eqref{xi-df}---this fact can be easily proven by setting the r.h.s. of \eqref{xiA-xiImu} to \eqref{xi-df} and expressing $\theta^0$ according to \eqref{th-bij-eq}. 

\subsection{Riemannian metrics on $\Sigma$}

According to the description of the phase space presented in the previous section each point $(\theta^A)$ of $\Theta$ defines the Riemannian metric \eqref{q} on $\Sigma$. Since $(\theta^J)$ is a global coframe on $\Sigma$ the metric can be expressed as      
\begin{equation}
q=q_{IJ}\theta^I\ot\theta^J
\label{q'}
\end{equation}
---the component $(q_{IJ})$ are obviously functions of $\theta^A$ but it turns out that they can be expressed as functions of space-like components of any solution $(\xi^A)$ of \eqref{xi-df} as well as by any variables $(\xi^I_\iota)$ \cite{q-suit}:   
\begin{equation}
q_{IJ}=\delta_{IJ}-\frac{\xi_{I}\xi_{J}}{1+\xi_{K}\xi^K}=\delta_{IJ}-\frac{\xi_{\iota I}\xi_{\iota J}}{1+\xi_{\iota K}\xi^K_\iota}.
\label{q-xi}
\end{equation}
It can be easily checked that a matrix $(\bar{q}^{IJ})$ inverse\footnote{Note that in general $q^{IJ}\neq \bar{q}^{IJ}$. This is because the components $q^{IJ}$ are obtained from $q_{IJ}$ by raising the indeces $I,J$ by means of $\delta^{IJ}$. The components $q^{IJ}$ and $\bar{q}^{IJ}$ coincide if and only if $\xi^I=0$.} to $(q_{IJ})$ reads
\begin{equation}
\bar{q}^{IJ}=\delta^{IJ}+\xi_\iota^I\xi_\iota^J.
\label{bar-qIJ}
\end{equation} 

Recall that $\Sigma$ is an oriented manifold. Therefore the metric $q$ defines a volume form  $\eps$ on $\Sigma$ and a Hodge operator $*$ acting on differential forms on the manifold. Obviously, both the volume form and the Hodge operator are functions of $(\theta^A)$ and, equivalently, $(\xi^I_\iota,\theta^J)$.

\subsection{Auxiliary formulae \label{aux}}

Here we will derive some formulae which will be used in calculations throughout the paper.     

We know from \cite{q-suit} that $(\theta^A)=(\theta^0,\theta^J)\in\Theta$ if and only if $(\theta^J)$ is a global coframe on $\Sigma$ and there exists a triplet $(\alpha_I)$ of real functions on the manifold such that $\alpha_I\alpha^I<1$ and 
\begin{equation}
\theta^0=\alpha_I\theta^I.
\label{t0-atI}
\end{equation}
Let us fix a function $\iota$ on $\Theta$ and consider the corresponding solution $(\xi^A_\iota)$ of \eqref{xi-df} . Then using results obtained in \cite{q-suit} we can express relations between the functions $\alpha_I$ and $(\xi^A_\iota)$ 
\begin{align}
&&(\xi_\iota^A)=(\xi_\iota^0,\xi_\iota^I)=\iota(\theta^J)\frac{(1,\alpha^I)}{\sqrt{1-\alpha_K\alpha^K}},&&\alpha_I=\iota(\theta^J)\frac{\xi_{\iota I}}{\sqrt{1+\xi_{\iota K}\xi_\iota^K}}=\frac{\xi_{\iota I}}{\xi_\iota^0},
\label{xi-al}
\end{align}
where $\xi^0_\iota$ is given by \eqref{xi-0}.

There hold the following formulae:
\begin{align}
\frac{\partial\alpha_I}{\partial\xi_\iota^J}&=\frac{1}{\xi_\iota^0}q_{IJ}, \label{al/xi}\\
*(\theta^1\we\theta^2\we\theta^3)&=\sgn(\theta^K)|\xi^0_\iota|, \label{*ttt}\\
*(\theta^I\we\theta^J)&=\sgn(\theta^N)|\xi^0_\iota|\,q_{LK}\veps^{IJL}\theta^K,\label{*tt}
\end{align}
---in \eqref{al/xi} $\alpha_I$ is given by the second equation in \eqref{xi-al}.

\begin{proof}[Proof of \eqref{al/xi}] Let us calculate
\begin{multline*}
\frac{\partial\alpha_I}{\partial\xi_\iota^J}=\iota(\theta^L)\Big(\frac{\delta_{IJ}}{\sqrt{1+\xi_{\iota K}\xi_\iota^K}}-\frac{\xi_{\iota I}\xi_{\iota J}}{\sqrt{1+\xi_{\iota K}\xi_\iota^K}^3}\Big)=\frac{\iota(\theta^L)}{\sqrt{1+\xi_{\iota K}\xi_\iota^K}}\Big(\delta_{IJ}-\frac{\xi_{\iota I}\xi_{\iota J}}{1+\xi_{\iota K}\xi_\iota^K}\Big)=\\=\frac{1}{\xi_\iota^0}\Big(\delta_{IJ}-\frac{\xi_{\iota I}\xi_{\iota J}}{1+\xi_{\iota K}\xi_\iota^K}\Big),
\end{multline*}
where in the last step we used \eqref{xi-0}. Now it is enough to apply \eqref{q-xi}.
\end{proof}

\begin{proof}[Proof of \eqref{*ttt}] The volume form $\eps$ can be expressed as
\[
\eps=\sgn(\theta^K)\sqrt{\det q_{IJ}}\,\theta^1\we\theta^2\we\theta^3.
\]   
It was shown in \cite{q-suit} that the eigenvalues of the matrix $(q_{IJ})$ are $1$, $1$ and $1-\alpha_K\alpha^K$, hence   
\[
\det(q_{IJ})=1-\alpha_K\alpha^K.
\]
It follows from the first equation in \eqref{xi-al}  that
\[
1-\alpha_K\alpha^K=\frac{1}{1+\xi_{\iota K}\xi_\iota^K}.
\]
By virtue of the last two expressions and \eqref{xi-0}
\begin{equation}
\xi^0_\iota=\frac{\iota(\theta^J)}{\sqrt{\det(q_{IJ})}}.
\label{xi0-detq}
\end{equation}
Consequently, 
\[
\eps=\sgn(\theta^K)\frac{1}{|\xi_\iota^0|}\,\theta^1\we\theta^2\we\theta^3.
\]
Acting on both sides of this equation by the Hodge operator $*$ and taking into account that $*\eps=1$ we obtain \eqref{*ttt}.
\end{proof}

\begin{proof}[Proof of \eqref{*tt}] We have
\[
\theta^I\we\theta^J=(\delta^I{}_M\delta^J{}_N-\delta^J{}_M\delta^I{}_N)\theta^N\ot\theta^M.
\] 
 Then
\begin{equation}
*(\theta^I\we\theta^J)=\frac{1}{2}(\delta^I{}_M\delta^J{}_N-\delta^J{}_M\delta^I{}_N)\bar{q}^{MM'}\bar{q}^{NN'}\eps_{M'N'K}\theta^K,
\label{*tt-0}
\end{equation}
where $(\eps_{IJK})$ are components of the volume form $\eps$ on $\Sigma$ in the coframe $(\theta^I)$. Note that
\[
\eps_{IJK}=\sgn(\theta^L)\sqrt{\det(q_{MN})}\,\veps_{IJK}.
\]
Using this equation we obtain 
\begin{multline*}
\bar{q}^{MM'}\bar{q}^{NN'}\eps_{M'N'K}=\bar{q}^{MM'}\bar{q}^{NN'}\bar{q}^{LL'}\eps_{M'N'L'}q_{LK}=\\=\sgn(\theta^{I'})\sqrt{\det(q_{IJ})}\bar{q}^{MM'}\bar{q}^{NN'}\bar{q}^{LL'}\veps_{M'N'L'}q_{LK}=\\=\sgn(\theta^{I'})\sqrt{\det(q_{IJ})}[\det(q_{IJ})]^{-1}\veps^{MNL}q_{LK}=\sgn(\theta^{I'})\frac{q_{LK}}{\sqrt{\det(q_{IJ})}}\veps^{MNL}.
\end{multline*}
Setting this result to \eqref{*tt-0} we get
\begin{multline*}
*(\theta^I\we\theta^J)=\frac{1}{2}\sgn(\theta^{L'})(\delta^I{}_M\delta^J{}_N-\delta^J{}_M\delta^I{}_N)\frac{q_{LK}}{\sqrt{\det(q_{I'J'})}}\veps^{MNL}\theta^K=\\=\sgn(\theta^{L'})\frac{q_{LK}}{\sqrt{\det(q_{MN})}}\veps^{IJL}\theta^K=\sgn(\theta^N)|\xi_\iota^0|\,q_{LK}\veps^{IJL}\theta^K,
\end{multline*}
where in the last step we used \eqref{xi0-detq}.
\end{proof}

\subsection{The new variables on $\Theta$ in terms of the natural ones}

The formula \eqref{th-bij-eq} describes the natural variables as functions of the new ones. Let us now inverse the formula. Clearly, to inverse it it is enough to express $\xi^I_\iota$ as a function of $(\theta^A)$.  

Using \eqref{xi-0} we can rewrite the transformation \eqref{th-bij-eq} in a more compact form
\begin{equation}
(\xi^I_\iota,\theta^J)\mapsto(\theta^0,\theta^J)=\Big(\frac{\xi_{\iota I}}{\xi_\iota^0}\theta^I,\theta^J\Big).
\label{t0-xitI}
\end{equation}
Consider now the following expression   
\begin{multline*}
\frac{1}{2}*(\theta^0\we\theta^I\we\theta^J)\,\veps_{IJK}=\frac{1}{2}\frac{\xi_{\iota L}}{\xi^0_\iota}*(\theta^L\we\theta^I\we\theta^J)\,\veps_{IJK}=\frac{1}{2}\frac{\xi_{\iota L}}{\xi^0_\iota}\veps^{LIJ}*(\theta^1\we\theta^2\we\theta^3)\,\veps_{IJK}=\\=\frac{1}{2}\frac{\xi_{\iota L}}{\xi_\iota^0}\sgn(\theta^N)\,|\xi_\iota^0|\,\veps^{LIJ}\veps_{IJK}=\frac{\sgn(\theta^N)}{\iota(\theta^J)}{\xi_{\iota L}}\delta^L{}_K=\frac{\sgn(\theta^N)}{\iota(\theta^J)}\xi_{\iota K}
\end{multline*}
---in these calculations above we used \eqref{t0-xitI} in the first step and \eqref{*ttt} in the third one. Thus we obtain the inverse of \eqref{th-bij-eq}
\begin{equation}
\begin{aligned}
\xi^K_{\iota}&=\frac{1}{2}\frac{\iota(\theta^L)}{\sgn(\theta^L)}*(\theta^0\we\theta_I\we\theta_J)\,\veps^{IJK},\\
\theta^J&=\theta^J. 
\end{aligned}
\label{xi-t}
\end{equation}

\section{Differentiability of new variables $(\xi^I_\iota,\theta^J)$ \label{diff}}

Let us recall some notions used in variational calculus. A curve in $\Theta$ 
\begin{equation}
]a,b[\,\ni\lambda\mapsto(\theta^A_\lambda)\in\Theta, \quad a<0<b,
\label{la->th}
\end{equation}
is differentiable if for every (local) coordinate chart $(y^i)$ on $\Sigma$ of a domain $U$ the following map
\[
]a,b[\,\times U\ni(\lambda,y)\mapsto \theta^A_{\lambda i}(y)\in \R
\]
where $(\theta^A_{\lambda i})$ are components of $\theta^A_\lambda$ in the coordinate chart $(y^i)$, is differentiable. Then we can say that a function $F:\Theta\to\R$ is differentiable at $(\theta^A)\in\Theta$ if 
\begin{enumerate}
\item for every curve \eqref{la->th} such that $(\theta^A_0)=(\theta^A)$ a map $\lambda\mapsto F(\theta^A_\lambda)$ is differentiable at $\lambda=0$; 
\item the variation
\[
\delta F(\theta^A):=\frac{d}{d\lambda}\Big|_{\lambda=0}F(\theta^A_\lambda)
\]    
is a linear function of the variation 
\[
\delta\theta^A:=\frac{d}{d\lambda}\Big|_{\lambda=0}\theta^A_\lambda.
\]
\end{enumerate}

In an analogous way we can define a differentiable curve
\begin{equation}
]a,b[\,\ni\lambda\mapsto(\xi^I_{\iota\lambda},\theta^J_\lambda)\in\Theta, \quad a<0<b,
\label{la->xith}
\end{equation}
and define a differentiability of the same function $F:\Theta\to\R$ by means of maps $\lambda\mapsto F(\xi^I_{\iota\lambda},\theta^J_\lambda)$.  

Thus we defined two notions of differentiability of $F$: one uses the natural variables and the other does the new ones. Of course, we would like both notions of differentiability to coincide---then we will say that the new variables are differentiable with respect to the natural ones (and vice versa). It is not difficult to realize that both notions coincide if $(i)$ the differentiability of \eqref{la->th} guarantees the differentiability of 
\begin{equation}
]a,b[\,\ni\lambda\mapsto\xi^K_{\iota\lambda}:=\frac{1}{2}\frac{\iota(\theta^L_\lambda)}{\sgn(\theta^L_\lambda)}*_\lambda(\theta^0_\lambda\we\theta_{I\lambda}\we\theta_{J\lambda})\,\veps^{IJK}
\label{la->xi}
\end{equation}
(see \eqref{xi-t}) and $(ii)$ if the differentiability of \eqref{la->xith}  guarantees the differentiability of
\begin{equation}
]a,b[\,\ni\lambda\mapsto\theta^0_\lambda:=\iota(\theta^L_\lambda)\frac{\xi_{\iota\lambda I}}{\sqrt{1+\xi_{\iota \lambda K}\xi_{\iota\lambda}^K}}\theta^I_\lambda.
\label{la->th0}
\end{equation}
(see \eqref{th-bij-eq}). Clearly, the only obstacle for the differentiability of \eqref{la->xi} and \eqref{la->th0} may be the function $\lambda\mapsto\iota(\theta^L_\lambda)$. Since $\iota$ is valued in the set $\{-1,1\}$ it is necessary and sufficient for the curves \eqref{la->xi} and \eqref{la->th0} to be differentiable to require that the function $\lambda\mapsto\iota(\theta^L_\lambda)$ is constant. Then the variations $(\delta\xi^I_\iota,\delta\theta^J)$ will be linear functions of $(\delta\theta^A)$ and vice versa.  

This result means that $(\xi^I_\iota,\theta^J)$ are differentiable with respect to the natural variables if $\iota$ is a constant function on every path-connected subset of $\Theta$---here a subset $\Theta_0\subset\Theta$ is path-connected if every pair of points of $\Theta_0$ can be connected by a path being a composition of finite number of differentiable curves \eqref{la->th}. Such a function $\iota$ will be called {\em admissible}. 

It is clear, that no point of $\Theta_+$ can be connected by such a path with any point of $\Theta_-$. This means that a set of all admissible  functions $\iota$ consists at least of four elements: $(i)$ a constant function equal $1$, $(ii)$ a constant function equal $-1$, $(iii)$ a function $\iota(\theta^J)=\sgn(\theta^J)$ and $(iv)$ a function $\iota(\theta^J)=-\sgn(\theta^J)$.           

Since now we will use merely differentiable variables $(\xi^I_\iota,\theta^J)$. Then $\iota$ is a constant function on every path-connected subset of $\Theta$ and therefore while  calculating any derivatives of formulae containing $\iota(\theta^I)$ we will treat it as a constant number.

\section{Constraints of TEGR and YMTM in terms of new variables \label{3+1}}

In this section we will express the natural momenta $(p_A)$ as function of the new variables $(\xi^I_\iota,\theta^J)$ and momenta conjugate to them. Then we will present the constraints of both TEGR and YMTM expressed in terms of the new variables on the phase space. Finally we will check which variables cannot be used to define quantum constraints.

\subsection{Momenta conjugate to new variables $(\xi^I_\iota,\theta^J)$ \label{new-mom}}

Since new variables $(\xi^I_\iota,\theta^J)$ on $\Theta$ are invertible and differentiable functions of the natural variables there should exist momenta conjugate to them. {Given function $\iota$ and value of the index $I$,} the variable $\xi^I_\iota$ is a zero-form on $\Sigma$ hence the momentum conjugate to it should be a three-form on the manifold which will be denoted by $\zeta_{\iota I}$---the exterior product of a variable and the momentum conjugate to it is a differential form of the degree equal the dimension of the manifold on which the forms are defined \cite{ham-diff,mielke} (see also \cite{os}). Analogously, the momentum  conjugate to $\theta^J$ should be a two-form on $\Sigma$ which will be denoted by $r_{ J}$.

Our goal now is to find a relation between the natural variables $(p_A,\theta^B)$ on the phase space and the new ones $(\zeta_{\iota I},r_{ J},\xi^K_\iota,\theta^L)$. Let
\begin{align*}
\zeta_{\iota I}&=a_I(p_A,\theta^B), & r_{ J}&=b_J(p_A,\theta^B)
\end{align*}    
To find the unknown functions $a_I$ and $b_J$ we could require that the functions together with \eqref{xi-t} define a canonical transformation on the phase space. This would give us partial differential equations imposed on $a_I$ and $b_J$. Alternatively, we could step back to the Lagrangian formulation of TEGR (and YMTM) and consider variables on the Lagrangian configuration space which define via the Legendre transformation the natural variables $(p_A,\theta^B)$ on the phase space. Then we could find new variables on the Lagrangian configuration space which define new variables $(\xi^I_\iota,\theta^J)$ on the Hamiltonian configuration space. Finally, comparing the Legendre transformation applied to both  sorts of variables we could find the functions $a_I$ and $b_J$. Since the former method requires solving partial differential equations and the latter one involves merely differentiation we will find $a_I$ and $b_J$ using the latter method.   

{Let us finally emphasize that in the considerations below we will repeatedly use the interior product\footnote{Let $\omega$ be a differential $k$-form and $X$ a vector field on a manifold. If $k>0$ then the interior product $X\lr\omega$ is a $(k-1)$-form such that for any vector fields $X_1,\ldots,X_{k-1}$ 
\[
(X\lr\omega)(X_1,\ldots,X_{k-1}):=\omega(X,X_1,\ldots,X_{k-1}),
\]      
if $k=0$ then $X\lr\omega:=0$.} $X\lr\omega$ of a vector field $X$ and a differential form $\omega$.}

\subsubsection{ADM-like variables on the Lagrangian configuration space}

Let $\cal M$ be a four-dimensional oriented. Let $(\bth^A)$, $A=0,1,2,3$, be a global {\em coframe} or a {\em cotetrad field} on the manifold compatible with its orientation. Then $\bth:=\bth^A\ot v_A$ is a one-form on $\cal M$ valued in $\mathbb{M}$ which can be used to pull back the scalar product $\eta$ on $\mathbb{M}$ to a Lorentzian metric on $\cal M$ turning thereby the manifold into a {\em spacetime}. The resulting metric $g$ reads
\begin{equation}
g:=\eta_{AB}\bth^A\ot\bth^B.
\label{g}
\end{equation}

The space of all such coframes $(\bth^A)$ is a Lagrangian configuration space for both\footnote{In fact, TEGR and YMTM can be defined on the space of all global cotetrad fields on $\cal M$ compatible and incompatible with the orientation---see e.g. \cite{mal-rev} and references therein. In \cite{oko-tegr} and \cite{os} we restricted ourselves to cotetrads compatible with the orientation because it simplified Hamiltonian formulations of the theories described in these papers.} TEGR and YMTM. 

To carry out the $3+1$ decomposition of the manifold $\cal M$  and a cotetrad field $(\bth^A)$ on it needed for the Legendre transformation we impose on them the following Assumptions:
\begin{enumerate}
\item ${\cal M}=\R\times\Sigma$. Since $\Sigma$ is oriented manifold we assume that the orientation of $\cal M$ is compatibles with the natural orientation of $\R\times\Sigma$.  
\item the cotetrad $(\bth^A)$ is such that for every $t\in\R$ the submanifold $\Sigma_t:=\{t\}\times\Sigma$ is space-like with respect to $g$.
\end{enumerate}

Assumption 1 allows us to introduce a family of curves in $\cal M$ parameterized by points of $\Sigma$---given $y\in\Sigma$ we define
\begin{equation}
\R\ni t\mapsto (t,y)\in\R\times\Sigma={\cal M}.
\label{curv}
\end{equation}
These curves generates a global vector field on $\cal M$ which will be denoted by $\partial_t$. By virtue of Assumption 1 there exists a real function on $t$ on $\cal M$ such that $t(y)=\tau$ if and only if $y\in\Sigma_{\tau}$. Moreover, Assumption 1 provides a family of natural embeddings $\varphi_t:\Sigma\to\Sigma_t\in{\cal M}$.

The space of all cotetrad fields $(\bth^A)$ on $\cal M$ which satisfy Assumption 2  will be called {\em the restricted Lagrangian configuration space} and will be denoted by $\bld{\Theta}$. 
Let us now describe variables which result from $3+1$ decomposition of the cotetrad fields, parameterize the space $\bld{\Theta}$ and lead to ADM-like Hamiltonian formulations of theories considered in \cite{os} and \cite{nester,oko-tegr} (the lemma below summarizes and makes more precise some facts described in \cite{nester,os,oko-tegr}; its proof can be found in Appendix \ref{app-TH-bij}):

\begin{lm}
There exists a one-to-one correspondence between elements of $\bld{\Theta}$ and all triplets $(N,\vec{N},\theta^A)$ consisting of 
\begin{enumerate}
\item a function $N$ on $\cal M$ which is positive everywhere;  
\item a vector field $\vec{N}$ on $\cal M$ tangent everywhere to the foliation $\{\Sigma_t\}_{t\in\R}$;   
\item one-forms $\theta^A$ ($A=0,1,2,3$) on $\cal M$ such that 
\begin{enumerate}
\item[a)] $\partial_t\lr\theta^A=0$; 
\item[b)] for every $t\in\R$ the tensor field
\[
q:=\eta_{AB}\theta^A\ot\theta^B
\] 
induces via the pull-back $\varphi^*_t$ a Riemannian (positive definite) metric on $\Sigma$.  
\end{enumerate}  
\end{enumerate}  
The correspondence is given by the following map
\begin{equation}
(N,\vec{N},\theta^A)\mapsto (\bth^A)=\Big(\bed t[-N\frac{1}{3!}\veps^A{}_{BCD}*(\theta^B\we\theta^C\we\theta^D)+\vec{N}\lr\theta^A]+\theta^A\Big)\in\bld{\Theta},
\label{NNth-bth}
\end{equation}
where $*$ is a Hodge operator acting on forms on $\Sigma_t$ defined by the induced metric on the manifold, and $\bed$ is the exterior derivative on $\cal M$.  
\label{TH-bij}
\end{lm}
\noindent Let us emphasize that desiring to keep the notation as simple as possible we did not decide to introduce a new symbol to denote the quadruplet of one-forms on $\cal M$ appearing in the lemma just stated and used the symbol $(\theta^A)$ which earlier denoted a quadruplet of one-forms on $\Sigma$. In the sequel the symbol will denote forms either on $\cal M$ or on $\Sigma$ (similarly, soon $(\xi^A)$ and $(\xi^I_\iota)$ will denote functions either on $\cal M$ or on $\Sigma$) and we hope that the meaning will be clear from the context.

As shown in \cite{os} ${N}$ and $\vec{N}$ coincide with, respectively, the ADM lapse and the ADM shift vector field \cite{adm}. It turns out \cite{oko-tegr,os} that actions of TEGR and YMTM do not contain time derivatives\footnote{Here the ``time derivative'' means the Lie derivative with respect to the vector field $\partial_t$.} of $N$ and $\vec{N}$ therefore the two variables can be treated as Lagrangian multipliers in the resulting Hamiltonian formulations. Moreover, it is clear that for every $t\in\R$ a quadruplet $(\theta^A)$ of one-forms on $\cal M$ described in Lemma \eqref{TH-bij} defines an element of the Hamiltonian configuration space $\Theta$ by means of the pull-back  with respect to the natural embedding $\varphi_t$. In other words, the quadruplet defines a curve in $\Theta$ parameterized by $t$. Obviously, this curve is differentiable in the sense described in Section \ref{diff}.   

The ADM-like variables $(N,\vec{N},\theta^A)$ on the restricted Lagrangian configuration space $\bld{\Theta}$ define via the Legendre transformation the natural variables $(p_A,\theta^A)$ on the phase space of TEGR and YMTM \cite{oko-tegr,os} together with Lagrangian multipliers $N$ and $\vec{N}$.        

\subsubsection{New variables on the restricted configuration space} 

Here we will introduce new variables on $\bld{\Theta}$ which, once the Legendre transformation has been carried out, define on the phase space the new variables $(\zeta_{\iota I},r_{ J},\xi^K_\iota,\theta^L)$. 

Let $(\theta^A)$ be a quadruplet of one-forms on $\cal M$ appearing in Lemma \ref{TH-bij}. As mentioned at the end of the previous section it defines a {\em differentiable} curve in the Hamiltonian configuration space $\Theta$:
\[
\R\ni t\mapsto(\varphi^*_t\theta^A)\in \Theta.
\]  
Therefore for every function $\iota$ which defines differentiable new variables on $\Theta$ the value $\iota(\varphi_t^*\theta^J)$ is independent of $t$. This value will be denoted by $\iota(\theta^J)$ since it characterizes the triplet $(\theta^J)$ of one-forms on $\cal M$. 

\begin{lm}
Given admissible function $\iota$ on $\Theta$, there exists a one-to-one correspondence between elements of $\bld{\Theta}$ and all quadruplets $(N,\vec{N},\xi_\iota^I,\theta^J)$ consisting of 
\begin{enumerate}
\item a function $N$ on $\cal M$ which is positive everywhere;  
\item a vector field $\vec{N}$ on $\cal M$ tangent everywhere to the foliation $\{\Sigma_t\}$;   
\item functions $\xi_\iota^I$ ($I=1,2,3$)  on $\cal M$, 
\item one-forms  $\theta^J$ ($J=1,2,3$) on $\cal M$ such that
\begin{enumerate}
\item[a)] $\partial_t\lr\theta^J=0$; 
\item[b)] for every $t\in\R$ the triplet $(\theta^J)$ induced via the pull-back $\varphi^*_t$ a global coframe on $\Sigma$. 
\end{enumerate}  
\end{enumerate}  
The correspondence is given by 
\begin{equation}
\begin{aligned}
(N,\vec{N},\xi_\iota^I,\theta^J)&\mapsto (\bth^A)=(\bth^0,\bth^I)\in\bld{\Theta},\\
\bth^0&=\bed t[N\sgn(\theta^J)\sqrt{1+\xi_{\iota K}\xi^K_\iota}+\vec{N}\lr\frac{\iota(\theta^J)\xi_{\iota I}}{\sqrt{1+\xi_{\iota K}\xi^K_\iota}}\theta^I]+\frac{\iota(\theta^J)\xi_{\iota I}}{\sqrt{1+\xi_{\iota K}\xi^K_\iota}}\theta^I,\\
\bth^I&=\bed t[N\sgn(\theta^J)\iota(\theta^J)\xi_\iota^I+\vec{N}\lr\theta^I]+\theta^I.
\end{aligned}
\label{NNxth-bth}
\end{equation}
\label{TH-bij-n}
\end{lm}
\noindent Let us note that this description of the space $\bld{\Theta}$ may be used e.g. to derive a Hamiltonian formulation of TEGR in a gauge which fixes values of the variables $(\xi^I_\iota)$---then the only ``position'' variable on the resulting phase space would be the global coframe $(\theta^J)$ on $\Sigma$.  

\begin{proof}
Note that due to Condition 4a of the lemma and Condition 3a of Lemma \ref{TH-bij} the one-forms $(\theta^I)$ and $(\theta^A)$ on $\cal M$ appearing in the lemmas can be restored from the families $(\varphi^*_t\theta^I)$ and $(\varphi^*_t\theta^A)$ of one-forms on $\Sigma$. Therefore we can use Lemma \ref{th-bij} to establish a one-to-one correspondence between the one-forms $(\theta^A)$ on $\cal M$ satisfying the requirements of Lemma \ref{TH-bij} and the fields $(\xi_\iota^I,\theta^J)$ on $\cal M$ satisfying the requirements of Lemma \eqref{TH-bij-n}. This correspondence is given by
\begin{equation}
(\xi^I_\iota,\theta^J)\mapsto \Big(\theta^0=\iota(\theta^L)\frac{\xi_{\iota I}}{\sqrt{1+\xi_{\iota K}\xi_{\iota}^K}}\theta^I,\theta^J\Big). 
\label{th-bij-M}
\end{equation}
To finish the proof it is enough to set in \eqref{NNth-bth} $(\theta^A)$ expressed in terms of $(\xi_\iota^I,\theta^J)$. The calculations are straightforward except the following ones:
\begin{multline}
-\frac{1}{3!}\veps^0{}_{BCD}*(\theta^B\we\theta^C\we\theta^D)=-\frac{1}{3!}\veps^0{}_{IJK}*(\theta^I\we\theta^J\we\theta^K)=\\=\frac{1}{3!}\veps_{IJK}\veps^{IJK}*(\theta^1\we\theta^2\we\theta^3)=\sgn(\theta^I)|\xi^0_\iota|=\sgn(\theta^I)\sqrt{1+\xi_{\iota K}\xi^K_\iota},
\label{xi-0s}
\end{multline}
where in the third step we used \eqref{*ttt} and in the last one \eqref{xi-0}. Similarly, by virtue of \eqref{th-bij-M}  
\begin{multline}
-\frac{1}{3!}\veps^I{}_{BCD}*(\theta^B\we\theta^C\we\theta^D)=-\frac{1}{2}\veps^I{}_{0JK}*(\theta^0\we\theta^J\we\theta^K)=\\=\frac{1}{2}\veps_0{}^I{}_{JK}\frac{\iota(\theta^N)\xi_{\iota L}}{\sqrt{1+\xi_{\iota M}\xi_\iota^M}}*(\theta^L\we\theta^J\we\theta^K)=\frac{1}{2}\veps^I{}_{JK}\iota(\theta^N)\xi_{\iota L}\veps^{LJK}\sgn(\theta^N)=\\=\sgn(\theta^N)\iota(\theta^N)\xi^I_{\iota}.
\label{xi-Is}
\end{multline}
\end{proof}

\subsubsection{New momenta as functions of the natural variables} 

Expressing the actions of TEGR and YMTM as functionals of the variables $(N,\vec{N},\xi_\iota^I,\theta^J)$ defined on $\cal M$ one can carry out the Legendre transformation and obtain momenta conjugate to the variables. The momenta conjugate to the lapse and the shift are zero because the functionals do not contain the time derivatives of the lapse and the shift---to see this recall that the actions expressed as functionals of $(N,\vec{N},\theta^A)$ do not contain the time derivatives of $N$ and $\vec{N}$ and these two variables do not appear in the relation \eqref{th-bij-M}. Thus again the lapse and the shift can be treated as Lagrange multipliers in the resulting Hamiltonian formulations and in this way one obtains the new variables $(\zeta_{\iota I},r_{ J},\xi^K_\iota,\theta^L)$ on the phase space. 
As stated at the beginning of Section \ref{new-mom} to find the new momenta $(\zeta_{\iota I},r_{ J})$ as functions of the natural variables $(p_A,\theta^A)$ on the phase space we will refer to the Legendre transformations. Let 
\begin{equation*}
L=L_1(\dot{\xi}_\iota^I,\dot{\theta}^J,\xi_\iota^K,\theta^L,N,\vec{N})=L_2(\dot{\theta}^A,\theta^B,N,\vec{N})
\end{equation*}
denote  a differential four-form {on the manifold $\cal M$} being the integrand of either the TEGR action or the YMTM action  expressed as $(i)$ a function $L_1$ of $(N,\vec{N},\xi_\iota^I,\theta^J)$ and their time derivatives (denoted by a dot over the symbol of a variable) and as $(ii)$ a function $L_2$ of $(N,\vec{N},\theta^A)$ and their time derivatives. 

Let $L_{a\perp}:=\partial_t\lr L_a$ ($a=1,2$). Since now we will treat all the variables $N,\vec{N},\xi_\iota^I,\theta^A$ and the time derivatives $\dot{\xi}_\iota^I,\dot{\theta}^A$ as time-dependent fields on $\Sigma$ defined appropriately via pull-backs\footnote{Note that the time derivative $\dot{\theta}^A$ can be restored from the family $\{\varphi^*_t\dot{\theta}^A\}$---indeed, the time derivative of $\theta^A$ being the Lie derivative of the one-form with respect to $\partial_t$ reads $\dot{\theta}^A=\partial_t\lr\bed \theta^A+\bed(\partial_t\lr\theta^A)=\partial_t\lr\bed\theta^A$, hence $\partial_t\lr\dot{\theta}^A=0$.} $\{\varphi^*_t\}$ and push-forwards $\{\varphi^{-1}_{t*}\}$. Then the momenta $(\zeta_{\iota I},r_{ J})$ conjugate to, respectively, $(\xi_\iota^I,\theta^J)$ and the momenta $(p_A)$ conjugate to $\theta^A$ can be defined as \cite{ham-diff,mielke,os}
\begin{align*}
&\zeta_{\iota I}:=\frac{\partial L_{1\perp}}{\partial \dot{\xi}_\iota^I},&&r_{ I}:=\frac{\partial L_{1\perp}}{\partial \dot{\theta}^J},&&p_A:=\frac{\partial L_{2\perp}}{\partial \dot{\theta}^A}.
\end{align*}
The partial derivative of a three-form with respect to an $l$-form is an $(3-l)$-form (for a definition of the partial derivative see e.g. \cite{os}). This means that $\zeta_{\iota I}$ is a three-form and $r_{ I}$ a two-form as stated already at the beginning of Section \ref{new-mom}.   

It turns out that to derive the desired relations it is more convenient to use standard description of the Legendre transformations and the momenta in terms of components of tensor densities expressed in a (local) coordinate frame $(y^i)$, $i=1,2,3$, on $\Sigma$. Let $\tilde{\eps}^{ijk}$ be a Levi-Civita density of weight $1$ on $\Sigma$. It allows to transform the three-forms $L_{1\perp}$ and $L_{2\perp}$ into scalar densities \cite{os}:
\begin{align*}
&\tilde{L}_1(\dot{\xi}_\iota^I,\dot{\theta}^J,\xi_\iota^K,\theta^L,N,\vec{N}):=\frac{1}{3!}(L_{1\perp})_{ ijk}\tilde{\eps}^{ijk},&&\tilde{L}_2(\dot{\theta}^A,\theta^B,N,\vec{N}):=\frac{1}{3!}(L_{2\perp})_{ ijk}\tilde{\eps}^{ijk}.
\end{align*}
Then
\begin{align*}
\tilde{\zeta}_{\iota I}:=\frac{\partial \tilde{L}_{1}}{\partial \dot{\xi}_\iota^I},&&\tilde{r}^{i}_{\iota I}:=\frac{\partial \tilde{L}_{1}}{\partial \dot{\theta}^J_i},&&\tilde{p}^i_A:=\frac{\partial \tilde{L}_{2}}{\partial \dot{\theta}^A_i},
\end{align*}
are tensor densities related to the momenta $\zeta_{\iota I},r_{ J},p_A$ as follows \cite{os}
\begin{equation}
\begin{aligned}
\zeta_{\iota I}&=\frac{1}{3!}\tilde{\zeta}_{\iota I}\,^{\text{\tiny{--1}}}\!\tl{\eps}_{ijk}\,dy^{i}\we dy^j\we dy^{k},&r_{ J}&=\frac{1}{2!}\tilde{r}^{i}_{\iota I}\,^{\text{\tiny{--1}}}\!\tl{\eps}_{ijk}\,dy^{j}\we dy^{k},\\p_A&=\frac{1}{2!}\tilde{p}^{i}_A\,^{\text{\tiny{--1}}}\!\tl{\eps}_{ijk}\,dy^{j}\we dy^{k},&&
\end{aligned}
\label{tm-m}
\end{equation}
where $\,^{\text{\tiny{--1}}}\!\tl{\eps}_{ijk}$ is the Levi-Civita density of weight $-1$ on $\Sigma$.   

Now let us find a relation between ``velocities'' $\dot{\theta}^0$ and $\dot{\xi}_\iota^I,\dot{\theta}^I$. By virtue of \eqref{t0-atI}
\begin{equation}
\dot{\theta}^0=\dot{\alpha}_I\theta^I+\alpha_I\dot{\theta}^I=\frac{\partial\alpha_I}{\partial\xi_\iota^J}\dot{\xi}_\iota^J\theta^I+\alpha_I\dot{\theta}^I=\frac{1}{\xi_\iota^0}(q_{IJ}\dot{\xi}_\iota^J\theta^I+\xi_{\iota I}\dot{\theta}^I),
\label{dot-t}
\end{equation}
where in the last step we applied \eqref{al/xi} and \eqref{xi-al}. 

This result together with the equation
\[
\tilde{L}_1(\dot{\xi_\iota^I},\dot{\theta}^J,\xi_\iota^I,\theta^J)=\tilde{L}_2\Big(\dot{\theta}^0(\dot{\xi_\iota^I},\dot{\theta}^J,\xi_\iota^I,\theta^J),\dot{\theta}^I,\theta^0(\xi_\iota^I,\theta^J),\theta^I\Big).
\]  
allows us express the momenta $\tilde{\zeta}_{\iota I}$ and $\tilde{r}^i_{\iota I}$ as functions of the variables $(\tilde{p}^i_{A},\theta^B_i)$. Thus
\begin{equation}
\tilde{\zeta}_{\iota I}=\frac{\partial\tilde{L}_2}{\partial\dot{\theta}^0_i}\frac{\partial\dot{\theta}^0_i}{\partial\dot{\xi_\iota^I}}=\frac{q_{IJ}}{\xi_\iota^0}\tilde{p}^i_{0}\,\theta^J_i.
\label{til-zeta}
\end{equation}
On the other hand 
\begin{equation}
\tilde{r}^i_{\iota I}=\frac{\partial\tilde{L}_2}{\partial\dot{\theta}^0_j}\frac{\partial \dot{\theta}^0_j}{\partial\dot{\theta}^I_i}+\frac{\partial\tilde{L}_2}{\partial\dot{\theta}^I_i}=\tilde{p}^i_0\frac{\xi_{\iota I}}{\xi_\iota^0}+\tilde{p}^i_I.
\label{til-r}
\end{equation}
Setting \eqref{til-zeta} to the first equation in \eqref{tm-m}  we obtain
\begin{multline*}
\zeta_{\iota I}=\frac{1}{3!}\tilde{\zeta}_{\iota I}\,^{\text{\tiny{--1}}}\!\tl{\eps}_{kln}\,dy^{k}\we dy^l\we dy^{n}=\frac{1}{3!}\frac{q_{IJ}}{\xi_\iota^0}\tilde{p}^j_{0}\,\theta^J_i\delta^i{}_j\,^{\text{\tiny{--1}}}\!\tl{\eps}_{kln}\,dy^{k}\we dy^l\we dy^{n}=\\=\frac{1}{3!}\frac{q_{IJ}}{\xi_\iota^0}\tilde{p}^j_{0}\,\theta^J_i\frac{1}{2}\,^{\text{\tiny{--1}}}\!\tl{\eps}_{jab}\tilde{\eps}^{iab}\,^{\text{\tiny{--1}}}\!\tl{\eps}_{kln}\,dy^{k}\we dy^l\we dy^{n}.
\end{multline*}
Since \cite{os}
\[
\tilde{\eps}^{iab}\,^{\text{\tiny{--1}}}\!\tl{\eps}_{kln}=3!\delta^{i}{}_{[k}\delta^a{}_l\delta^{b}{}_{n]}
\]
we can write
\begin{multline}
\zeta_{\iota I}=\frac{1}{3!}\frac{q_{IJ}}{\xi_\iota^0}\tilde{p}^j_{0}\,\theta^J_i\frac{1}{2}\,^{\text{\tiny{--1}}}\!\tl{\eps}_{jab}\,3!\delta^i{}_k\delta^a{}_l\delta^b{}_n\,dy^{k}\we dy^l\we dy^{n}=\\=\frac{q_{IJ}}{\xi_\iota^0}\,\theta^J_i\,dy^{i}\we \Big(\frac{1}{2}\tilde{p}^j_{0} \,^{\text{\tiny{--1}}}\!\tl{\eps}_{jab} dy^a\we dy^{b}\Big)=
\frac{q_{IJ}}{\xi_\iota^0}\,\theta^J\we p_0,
\label{zeta-f}
\end{multline}
where in the last step we applied the third equation in \eqref{tm-m}. Similarly, setting \eqref{til-r} to the second equation in \eqref{tm-m} and using the third equation in \eqref{tm-m} we get
\begin{equation}
r_I=\frac{\xi_{\iota I}}{\xi_\iota^0}p_0+p_I.
\label{r-f}
\end{equation} 

Now, we are going to inverse the formulae \eqref{zeta-f} and \eqref{r-f}. {Denoting by $(p_{0MN})$ the components of $p_0$ given by the coframe $(\theta^I)$ we transform the former formula as follows:} 
\begin{multline*}
*\zeta_{\iota I}=\frac{q_{IJ}}{\xi_\iota^0}*(\theta^J\we p_0)=\frac{q_{IJ}}{\xi_\iota^0}\frac{1}{2}p_{0MN}*(\theta^J\we \theta^M\we\theta^N)=\\=\frac{q_{IJ}}{\xi_\iota^0}\frac{1}{2}p_{0MN}\veps^{JMN}*(\theta^1\we \theta^2\we\theta^3)=\frac{q_{IJ}}{\xi_\iota^0}\frac{1}{2}p_{0MN}\veps^{JMN}\sgn(\theta^K)|\xi_\iota^0|=\\=q_{IJ}\frac{\sgn(\theta^K)}{\iota(\theta^K)}\frac{1}{2}p_{0MN}\veps^{JMN}.
\end{multline*}
Thus
\[
\frac{\iota(\theta^K)}{\sgn(\theta^K)}\bar{q}^{JI}*\zeta_{\iota I}=\frac{1}{2}p_{0MN}\veps^{JMN},
\]
where $(\bar{q}^{IJ})$ is the inverse of $(q_{IJ})$ given by \eqref{bar-qIJ}. Contracting both sides of this equation with the permutation symbol $\veps_{JM'N'}$ we obtain  
\[
\frac{\iota(\theta^K)}{\sgn(\theta^K)}\bar{q}^{JI}*\zeta_{\iota I}\veps_{JM'N'}=\frac{1}{2}p_{0MN}\veps^{JMN}\veps_{JM'N'}=p_{0M'N'}.
\]
Consequently,
\[
p_0=\frac{1}{2}p_{0MN}\theta^M\we\theta^N=\frac{1}{2}\frac{\iota(\theta^K)}{\sgn(\theta^K)}\bar{q}^{JI}(*\zeta_{\iota I})\veps_{JMN}\theta^M\we\theta^N.
\]
Let us simplify the result: acting on both sides of this equation by the Hodge operator $*$ and using \eqref{*tt} we obtain 
\begin{multline*}
*p_0=\frac{1}{2}\frac{\iota(\theta^K)}{\sgn(\theta^K)}\bar{q}^{JI}(*\zeta_{\iota I})\veps_{JMN}*(\theta^M\we\theta^N)=\\=\frac{1}{2}\iota(\theta^{K'})\bar{q}^{JI}(*\zeta_{\iota I})\,\veps_{JMN}\,|\xi^0_\iota|\,q_{LK}\veps^{MNL}\theta^K=\\=\xi_\iota^0(*\zeta_{\iota I})\bar{q}^{JI}\delta^L{}_Jq_{LK}\theta^K=\xi_\iota^0(*\zeta_{\iota I})\theta^I.
\end{multline*}
Thus
\begin{equation}
p_0=**p_0=\xi_\iota^0*(*\zeta_{\iota I}\we\theta^I)=\xi_\iota^0\,\vec{\theta}^I\lr\zeta_{\iota I},
\label{p0-tze}
\end{equation}
---here in the last step we used an identity \cite{os}:
\begin{equation}
*(*\beta\we\alpha)=\vec{\alpha}\lr\beta,
\label{a-b}
\end{equation}
valid for any $k$-form $\beta$ and any one-form $\alpha$ on $\Sigma$, where $\vec{\alpha}$ denotes a vector field obtained from the one-form $\alpha$ by raising its index by the metric inverse to $q$: in a coordinate frame on $\Sigma$ 
\[
\vec{\alpha}:=q^{ij}\alpha_i\partial_j.
\]   

Setting \eqref{p0-tze} to \eqref{r-f} we obtain
\begin{equation}
p_I=r_{ I}-\xi_{\iota I}\,\vec{\theta}^J\lr\zeta_{\iota J}.
\label{pI-rtze}
\end{equation}

\subsubsection{Summary of the transformations}

Gathering \eqref{p0-tze}, \eqref{pI-rtze} and \eqref{th-bij-eq}  and using \eqref{xi-0} we obtain the following formulae describing the dependence of $(p_A,\theta^B)$ on $(\zeta_{\iota I},r_{ J},\xi_\iota^K,\theta^L)$:
\begin{equation}
\begin{aligned}
p_0&=\iota(\theta^K)\sqrt{1+\xi_{\iota J}\xi^J_\iota}\,\vec{\theta}^I\lr\zeta_{\iota I},& p_I&=r_I-\xi_{\iota I}\,\vec{\theta}^J\lr\zeta_{\iota J},\\
\theta^0&=\iota(\theta^J)\frac{\xi_{\iota I}}{\sqrt{1+\xi_{\iota K}\xi_\iota^K}}\,\theta^I, &  \theta^I&=\theta^I
\end{aligned} 
\label{old-new}
\end{equation}
---here the metric inverse to $q$ used to define the vector field $\vec{\theta}^I$ should be treated as a function of $\xi_\iota^I$ and $\theta^J$ (see Equations \eqref{q'} and \eqref{q-xi}).

Gathering \eqref{zeta-f}, \eqref{r-f} and \eqref{xi-t} and using \eqref{xi0-detq} we obtain the following formulae expressing the dependence of $(\zeta_{\iota I},r_{ }J,\xi_\iota^K,\theta^L)$ on $(p_A,\theta^B)$:
\begin{equation}
\begin{aligned}
\zeta_{\iota I}&=\iota(\theta^K)\sqrt{\det (q_{MN})}\,q_{IJ}\,\theta^J\we p_0, \\
r_{ I}&= \frac{\sqrt{\det (q_{MN})}}{2}\sgn(\theta^L)*(\theta^0\we\theta^J\we\theta^K)\,\veps_{JKI}\,p_0+p_I,\\
\xi^I_{\iota}&=\frac{1}{2}\frac{\iota(\theta^L)}{\sgn(\theta^L)}*(\theta^0\we\theta_J\we\theta_K)\,\veps^{JKI},\\
\theta^I&=\theta^I.
\end{aligned}
\label{new-old}
\end{equation}
In these formulae the Hodge operator $*$ is defined by the metric $q$ treated as a function \eqref{q} of $(\theta^A)$. Note that it follows immediately from the result just obtained that $r_I$ does not depend on the function $\iota$. 

Let $\iota_1$ and $\iota_2$ be admissible functions on $\Theta$. Then $\iota:=\iota_1/\iota_2$ is admissible as well and 
\begin{align*}
\zeta_{\iota_1 I}&=\iota \zeta_{\iota_2 I},&\xi^I_{\iota_1}&=\iota\xi^I_{\iota_2}.
\end{align*}

Analyzing the formulae \eqref{new-old} we will find now the range of the new momenta $\zeta_{\iota I}$ and $r_J$ (the range of $(\xi^K_\iota,\theta^L)$ is described by Lemma \ref{th-bij}). Since $p_I$ can be any two-form on $\Sigma$ the momentum $r_I$ can also be any two-form on the manifold. Consider now a three-form $(\bar{q}^{JI}\zeta_{\iota I})$ which according to the first equation in \eqref{new-old} is of the following form     
\begin{equation}
\bar{q}^{JI}\zeta_{\iota I}=\theta^J\we p,
\label{q-zeta}
\end{equation}
where $p$ can be any two-form on $\Sigma$. On the other hand any three-form $\alpha$ on $\Sigma$ can be expressed as
\[
\alpha=\frac{1}{3!}\alpha_{IJK}\theta^I\we\theta^J\we\theta^K=\alpha_{123}\theta^1\we\theta^2\we\theta^3.
\]  
It means that there are no restriction imposed on the form $\bar{q}^{JI}\zeta_{\iota I}$---indeed, if  e.g. $J=1$ then setting in \eqref{q-zeta} $p=\alpha_{123}\theta^2\we\theta^3$ we see that $\bar{q}^{1I}\zeta_{\iota I}=\alpha$. Since there are no restrictions imposed on $\bar{q}^{JI}\zeta_{\iota I}$ there are no restrictions  imposed on $\zeta_{\iota I}$.      
  
Thus we obtain a new description of the phase space $P\times\Theta$ alternative to that presented in Section \ref{ph-sp}: given admissible function $\iota$, 
\begin{enumerate}
\item $\Theta$ consists of all sextuplets $(\xi_\iota^I,\theta^J)$, $(I,J=1,2,3)$, such that $\xi_\iota^I$ is a real function on $\Sigma$ and the one-forms $(\theta^J)$ form a global coframe on the manifold (see Lemma \ref{th-bij})      
\item $P$ consists of all sextuplets $(\zeta_{\iota I},r_J)$, $(I,J=1,2,3)$, where $\zeta_{\iota I}$ is a three-from and $r_J$ a two-form on $\Sigma$.     
\end{enumerate}              
The Poisson bracket on the phase space in terms of the new variables reads:
\[
\{F,G\}=\int_\Sigma\Big(\frac{\delta F}{\delta {\xi}_\iota^I}\we\frac{\delta G}{\delta \zeta_{\iota I}}+\frac{\delta F}{\delta {\theta}^I}\we\frac{\delta G}{\delta r_I}-\frac{\delta G}{\delta {\xi}_\iota^I}\we\frac{\delta F}{\delta \zeta_{\iota I}}-\frac{\delta G}{\delta {\theta}^I}\we\frac{\delta F}{\delta r_I}\Big).
\]

\subsection{Constraints of TEGR and YMTM}

\subsubsection{The constraints as functions of the natural variables \label{constr-old}}

Let us first express the (smeared) constraints (and the Hamiltonians) of TEGR and YMTM in terms of the natural variables $(p_A,\theta^B)$ and the function $(\xi^A)$ given by \eqref{xi-*} (being thereby a function of $(\theta^A)$). 

In \cite{oko-tegr} we derived a complete set of constraints of TEGR consisting of a scalar constraint
\begin{multline}
S(M):=\int_\Sigma M\Big(\frac{1}{2}(p_A\we\theta^B)\we*(p_B\we\theta^A)-\frac{1}{4}(p_A\we\theta^A)\we*(p_B\we\theta^B)-\xi^A{d}p_A+\\+\frac{1}{2}(d\theta_A\we\theta^B)\we{*}(d\theta_B\we\theta^A)-\frac{1}{4}(d\theta_A\we\theta^A)\we{*}(d\theta_B\we\theta^B)\Big)
\label{S(N)}
\end{multline}
and, respectively, smeared vector, boost and rotation constraints (the last two constraints generate local Lorentz transformations on the phase space):
\begin{align}
V({\vec{M}}):=&\int_\Sigma -{d}{\theta}^A\we(\vec{M}\lrcorner p_A)-(\vec{M}\lrcorner{\theta}^A)\we {d}p_A, \label{V(N)}\\
B(a):=&\int_\Sigma a\we(\theta^A\we*d\theta_A+\xi^Ap_A),\label{B(a)}\\
R(b):=&\int_\Sigma b\we(\theta^A\we*p_A-\xi^Ad\theta_A),\label{R(b)}.
\end{align}
In the formulae above there appear the following smearing fields on $\Sigma$: $M$ is a functions, $\vec{M}$ is a vector field and $a$ and $b$ are one-forms. All the constraints are of the first class.  

The Hamiltonian of TEGR turns out to be a sum of the constraints:   
\begin{equation}
H[{\theta}^A,p_B,N,\vec{N},a,b]=S(N)+V(\vec{N})+B(a)+R(b),
\label{full-ham-c}
\end{equation}
where $N$ is the lapse function, $\vec{N}$ is the shift vector field---here the fields $N,\vec{N},a$ and $b$ Lagrangian multipliers. 

The only constraints of YMTM \cite{os} are a scalar constraint
\begin{equation}
s(M):=\int_\Sigma M\Big(\frac{1}{2}p^A\we{*}p_A-\xi^Adp_A+\frac{1}{2}{d}{\theta}^A\we{*}{d}{\theta}_A\Big)
\label{s(N)}
\end{equation}
and the vector constraint $v(\vec{M})\equiv V(\vec{M})$. The constraints are of the first class.

The Hamiltonian of YMTM is of the following form 
\begin{equation}
h[{\theta}^A,p_A,N,\vec{N}]=s(N)+v(\vec{N})
\label{ham-simpl}
\end{equation}

\subsubsection{The constraints as functions of new variables \label{constr-new}}

To rewrite the constraints in terms of new variables $(\zeta_{\iota I},r_J,\xi^K_\iota,\theta^L)$ it is enough to set in the formulae presented above $(p_A,\theta^B)$ and $(\xi^A)$ expressed as functions of the new variables, that is, \eqref{old-new}, \eqref{xi-0s} and \eqref{xi-Is} (recall that $(\xi^A)$ appearing in the constraints is defined by \eqref{xi-*}). For the sake of simplicity wherever possible we will use the function $\xi^0_\iota$ defined by \eqref{xi-0}.     

Calculations needed to transform the constraints of TEGR and YMTM to the desired form will be carried out in Appendix \ref{constr-der}, here we only present the results. 

The scalar constraints of TEGR reads
\begin{multline}
S(M)=\int_\Sigma M\Big(\frac{1}{2}(r_I\we\theta^J)\we*(r_J\we\theta^I)-\frac{1}{4}(r_I\we\theta^I)\we*(r_J\we\theta^J)-\\-\frac{\sgn(\theta^L)}{\iota(\theta^L)}\big(d(\vec{\theta}^J\lr\zeta_{\iota J})+\xi_\iota^I\we dr_I\big)+\\+\frac{1}{4(\xi_\iota^0)^4}d(\xi_{\iota I}\theta^I)\we\xi_{\iota J}\theta^J\we*(d(\xi_{\iota K}\theta^K)\we\xi_{\iota L}\theta^L)+\frac{1}{2(\xi_\iota^0)^2}d(\xi_{\iota I}\theta^I)\we\xi_{\iota J}\theta^J\we*(d\theta_K\we\theta^K)-\\-\frac{1}{(\xi_\iota^0)^{2}}(q_{IJ}d\xi_\iota^J\we\theta^I+\xi_{\iota I}d\theta^I)\we\theta^K\we*(d\theta_K\we\xi_{\iota L}\theta^L)+\\+\frac{1}{2}d\theta_I\we\theta^J\we*(d\theta_J\we\theta^I)-\frac{1}{4}d\theta_I\we\theta^I\we*(d\theta_J\we\theta^J)\Big).
\label{scal-c}
\end{multline}
The vector constraint
\begin{equation}
V(\vec{M})=\int_\Sigma d\xi_\iota^I\we\vec{M}\lr\zeta_{\iota I} -{d}{\theta}^I\we(\vec{M}\lrcorner\, r_I)-(\vec{M}\lrcorner{\theta}^I)\we {d}r_I.
\label{v(n)-new}
\end{equation}
The boost constraint
\begin{multline}
B(a)=\\=\int_\Sigma a\we\Big(-\frac{\xi_{\iota I}q_{JK}}{(\xi_\iota^0)^2}\theta^I\we*(d\xi_\iota^J\we\theta^K)+q_{IJ}\theta^I\we*d\theta^J+\frac{\sgn(\theta^L)}{\iota(\theta^L)}\big(\vec{\theta}^I\lr\zeta_{\iota I}+\xi_\iota^I r_I\big)\Big).
\label{boost-c}
\end{multline}
The rotation constraint
\begin{equation}
R(b)=\int_\Sigma b\we(\theta^I\we*r_I+\frac{\sgn(\theta^L)}{\iota(\theta^L)}q_{IJ}d\xi_\iota^I\we\theta^J).
\label{rot-c}
\end{equation}

The scalar constraints of YMTM reads:
\begin{multline}
s(M)=\int_\Sigma \frac{M}{2}\Big(-\bar{q}^{IJ}\zeta_{\iota I}*\zeta_{\iota J}+r^I\we*r_I-2\xi_\iota^I*\zeta_{\iota K}\, r_I\we{\theta}^K +\\-2\frac{\sgn(\theta^L)}{\iota(\theta^L)}\big(d(\vec{\theta}^J\lr\zeta_{\iota J})+\xi_\iota^I\we dr_I\big)-\\-\frac{q_{IJ}q_{KL}}{(\xi^0_\iota)^2}d\xi_\iota^I\we\theta^J\we*(d\xi_\iota^K\we\theta^L)-\frac{2q_{IJ}}{(\xi^0_\iota)^2}d\xi_\iota^I\we\theta^J\we*(\xi_{\iota K}d\theta^K)+q_{IJ}d\theta^I\we*d\theta^J\Big).
\label{scal-c-ymtm}
\end{multline}
The vector constraint $v(\vec{M})$ of YMTM coincides with that of TEGR. 

\subsection{An obstacle for defining quantum constraints}

\subsubsection{Outline of the construction of a space of quantum states \label{outline}}

To check for which new variables $(\zeta_{\iota I},r_J,\xi^K_\iota,\theta^L)$ there appears an obstacle for defining quantum constraints let us first outline the projective methods \cite{q-nonl} by means of which we would like to construct a space of kinematic quantum states for TEGR. 

The methods require to choose some functions on the Hamiltonian configuration space $\Theta$ as well as some functions on the momentum space $P$---the former functions are called {\em configurational elementary degrees of freedom}, while the latter ones {\em momentum elementary d.o.f.}. All the elementary d.o.f. should separate points in the phase space. Moreover, it should be possible to construct from the d.o.f. a directed set $(\Lambda,\geq)$ such that each element $\lambda$ of this set corresponds to a finite number of both momentum and configurational d.o.f.. It was shown in \cite{q-nonl} that if the set $(\Lambda,\geq)$ satisfies some assumptions then it naturally generates a space $\D$ of kinematic quantum states. The space $\D$ is generated in the following way.

Given $\lambda\in\Lambda$ which corresponds to a finite set $K=\{\kappa_1,\ldots,\kappa_N\}$ of configurational d.o.f. (and to a finite set of momentum ones), one defines so called {\em reduced configuration space} $\Theta_K$:
\[
\Theta_K:=\Theta/\sim_K,
\]         
where $\sim_K$ is an equivalence relation on $\Theta$---we say that $\theta,\theta'\in\Theta$ are equivalent if $\kappa_\alpha(\theta)=\kappa_\alpha(\theta')$ for every $\kappa_\alpha\in K$. One of the assumptions imposed on $(\Lambda,\geq)$ requires that there exists a natural bijection from $\Theta_K$ onto $\R^N$, where $N$ is the number of elements of $K$. This allows to define a Hilbert space
\[
\h_\lambda:=L^2(\Theta_K,dx),
\]     
where $dx$ is a measure on $\Theta_K$ induced by the Lebesgue measure on $\R^N$ via the natural bijection. Then among all bounded operators on $\h_\lambda$ one distinguishes the space $\D_\lambda$ of all density operators on $\h_\lambda$ (i.e. positive operators on the Hilbert space of trace equal $1$). In this way one obtains a family $\{\D_\lambda\}_{\lambda\in\Lambda}$. It follows from the assumptions the set $(\Lambda,\geq)$ is supposed to satisfy that this family is naturally equipped with the structure of a projective family. Then the space $\D$ of quantum states is defined as the projective limit of the family. 

\subsubsection{How to apply the construction to TEGR?}
        
Let us now explain how we are going to apply this general construction to TEGR. Let $y$ be a point of $\Sigma$, $e$ an edge\footnote{A {\em simple edge} is a one-dimensional connected $C^\infty$ submanifold of $\Sigma$ with two-point boundary. An edge is an {\em oriented} one-dimensional connected $C^0$ submanifold of $\Sigma$ given by a finite union of simple edges.} in the manifold and $\iota$ an admissible function. Consider the following functions on $\Theta$ \cite{q-suit}:
\begin{equation}
\begin{aligned}
\Theta\ni\theta\mapsto\kappa^I_y(\theta)&:=\xi^I_\iota(y)\in\R,\\
\Theta\ni\theta\mapsto\kappa^J_e(\theta)&:=\int_e\theta^J\in\R.
\end{aligned}
\label{con-dof}
\end{equation}
It was shown in \cite{q-suit} that all such functions (where $I,J=1,2,3$, $y$ runs through $\Sigma$ and $e$ through all edges in the manifold) are very promising as configurational elementary d.o.f. for constructing a set $(\Lambda,\geq)$ and  thereby a space $\D$ of quantum states for TEGR. More precisely, we argued there that to construct the directed set one should use finite sets of configurational d.o.f. of the following form
\begin{equation}
{K}_{u,\gamma}:=\{ \ {\kappa}^I_{y_1},\ldots,{\kappa}^I_{y_M},\kappa^J_{e_1},\ldots,\kappa^J_{e_N} \ | \ I,J=1,2,3 \ \}.
\label{K-ug1}
\end{equation}
where $u=\{y_1,\ldots,y_M\}$ is a finite subset of $\Sigma$ and $\gamma=\{e_1,\ldots,e_N\}$ is a graph\footnote{We say that two edges are {\em independent} if the set of their common points is either empty or consist of one or two endpoints of the edges. A {\em graph} in $\Sigma$ is a finite set of pairwise independent edges.} in the manifold. As proven in \cite{q-suit} there exists a natural bijection from the  reduced configuration space $\Theta_{K_{u,\gamma}}$ onto  $\R^{3(M+N)}$. Thus one can try to build the space $\D$ for TEGR from the spaces $\{\D_\lambda\}_{\lambda\in\Lambda}$ of density operators each associated with a Hilbert space of square-integrable functions on some $\Theta_{K_{u,\gamma}}$.        

\subsubsection{The obstacle}

Suppose that, given new variables $(\zeta_{\iota I},r_J,\xi^K_\iota,\theta^L)$, there exists a space $\D$ of quantum states for TEGR constructed as explained above\footnote{In \cite{q-stat} we will show that actually this supposition {\em is true}.}.  Given constraint $C$ on the phase space, we may to try to define its quantum counterpart $\hat{C}$ on $\D$ as a family $\{\hat{C}_\lambda\}_{\lambda\in\Lambda}$ of operators such that each $\hat{C}_\lambda$ is a quantum constraint on $\h_\lambda$ \cite{q-nonl}---taking into account the complexity of the whole space $\D$ it would be rather very difficult or perhaps impossible to define quantum constraints via an essentially different method. 

The question now is: are we able to define operators $\{\hat{C}_\lambda\}$ for the constraints of TEGR (or YMTM)? Assume that $\iota(\theta^J)\neq\sgn(\theta^J)$ and $\iota(\theta^J)\neq-\sgn(\theta^J)$. Then all the constraints of TEGR and YMTM except the vector one depend on $\sgn(\theta^J)$. This means that to define an operator $\hat{C}_\lambda$ corresponding to a constraint under consideration  we would have to represent the function $\sgn(\theta^I)$ as an operator on $\h_\lambda$. This however seems to be impossible.    

Indeed, there holds the following lemma \cite{q-suit}:
\begin{lm}
Let $\gamma=\{e_1,\ldots,e_N\}$ be a graph. Then for every $(x^I_{\bar{J}})\in\R^{3N}$ there exists a global coframe $(\theta^I)$ on $\Sigma$ compatible (incompatible) with the orientation of the manifold such that  
\[
\int_{e_{\bar{J}}}\theta^I=x^I_{\bar{J}}
\]      
for every $I=1,2,3$ and $\bar{J}=1,2,\ldots,N$. 
\label{theta-x3}
\end{lm}
\noindent By virtue of the lemma for every $\theta\equiv(\xi_\iota^I,\theta^J)\in\Theta$ the equivalence class $[\theta]\in \Theta_{K_{u,\gamma}}$ defined by the relation $\sim_{K_{u,\gamma}}$ contains points of $\Theta$ given by global coframes on $\Sigma$ compatible and global coframes incompatible with the orientation of the manifold. Therefore no function on $\Theta_{K_{u,\gamma}}$ can even approximate the function $\sgn(\theta^I)$. 

Of course, all d.o.f. \eqref{con-dof} (or even countably infinite subset of all these d.o.f.) contain enough information to obtain the value of $\sgn(\theta^I)$ for any global coframe $(\theta^I)$. However, this fact is rather not very helpful since it means that we would have to define some quantum constraints directly on $\D$ or on sectors of the space such that each sector is given by an infinite number of quantum d.o.f. corresponding to the classical ones and, of course, this task seems to be very hard.     

The conclusion is that, unless $\iota(\theta^J)=\sgn(\theta^J)$ or $\iota(\theta^J)=-\sgn(\theta^J)$, for most of the constraints we cannot find quantum counterparts via the method described above which seems to be only workable one. In other words, even if for an admissible $\iota(\theta^J)\neq\pm\sgn(\theta^J)$ one could construct the space $\D$ then it would not be useful in the sense described in the introduction to the paper.    

Obviously, the two admissible functions $\sgn(\theta^J)$ and $-\sgn(\theta^J)$ are distinguished because for the function $(\xi^A)$ given by \eqref{xi-*} appearing in the original form of the constraints 
\[
\sgn(\xi^0)=\sgn(\theta^J)
\]
---see \eqref{xi-0s}. Let us denote by 
\begin{align*}
&(\zeta_{s I},r_J,\xi^K_s,\theta^L), && \big((\zeta_{-s I},r_J,\xi^K_{-s},\theta^L)\big)
\end{align*}
the new variables defined by $\sgn(\theta^J)$ ($-\sgn(\theta^J)$). Now the conclusion can be rephrased as follows: the variables $(\zeta_{\pm s I},r_J,\xi^K_{\pm s},\theta^L)$ are the only new variables on the phase space for which the obstacle considered in this section does not appear.      
    
\section{Summary}

In this paper we proceeded further with the analysis of the family of new variables $\{(\xi^I_\iota,\theta^J)\}$ which are promising for a construction of the space of kinematic quantum states for TEGR and YMTM via the projective methods described in \cite{q-nonl}. In particular,  
\begin{enumerate}
\item we found a criterion which distinguishes differentiable variables in the family: new variables $(\xi^I_\iota,\theta^J)$ are differentiable if the function $\iota$ defining the variables is constant on path-connected subsets of the Hamiltonian configuration space $\Theta$;   
\item for every differentiable variables $(\xi^I_\iota,\theta^J)$ we derived conjugate momenta $(\zeta_{\iota I},r_J)$ as functions of the natural variables $(p_A,\theta^B)$; we also found formulae describing the dependence of $(p_A,\theta^B)$ on $(\zeta_{\iota I},r_J,\xi^K_\iota,\theta^L)$;   
\item we expressed the constraints (and thereby the Hamiltonians) of TEGR and YMTM in terms of new variables $(\zeta_{\iota I},r_J,\xi^K_\iota,\theta^L)$ on the phase space;
\item we showed that for all new variables on the phase space except  $(\zeta_{\pm s I},r_J,\xi^K_{\pm s},\theta^L)$ given by the functions $\iota(\theta^J)=\pm\sgn(\theta^J)$ there appears an obstacle which makes very difficult (if not impossible) the task of defining quantum constraints on the resulting space of kinematic quantum states.   
\end{enumerate}

In \cite{q-stat} we will construct the desired space $\D$ of kinematic quantum states for TEGR using the variables  $(\zeta_{s I},r_J,\xi^K_{s},\theta^L)$. Then we will show that $(\zeta_{-s I},r_J,\xi^K_{-s},\theta^L)$ define the same space $\D$. It will also become obvious that every other new variables $(\zeta_{\iota I},r_J,\xi^K_{\iota},\theta^L)$ define a space of kinematic quantum states which, however, does not seem to be useful for further stages of the quantization based on the Hamiltonian formulation derived in \cite{oko-tegr}. 

Finally let us comment on the constraints of TEGR expressed in terms of new variables in Section \ref{constr-new}. Evidently, the formulae describing the constraints became more complicated in comparison to the original version of the constraints (Section \ref{constr-old}) and this may cause a feeling that it will be much more troublesome to impose  quantum constraints $\{\hat{C}_\lambda\}$ in terms of new variables than in term of the natural ones. However, it is not necessary the case. 

First of all, perhaps it is possible to rewrite the constraints in a simpler way. If not, then note that what really makes the new version of the constraints more complicated are some terms containing the variables $(\xi^I_{\iota})$ and their functions like $\xi^0_\iota$ (see \eqref{xi-0}) or the components $(q_{IJ})$ (see \eqref{q-xi}). Denoting respectively by $x^I_y$ and $x^J_e$ the values of the maps \eqref{con-dof} it is easy to see that 
\[
(x^I_{y_1},\ldots,x^I_{y_M},x^J_{e_1},\ldots,x^J_{e_N}), \quad I,J=1,2,3, 
\]          
are global coordinates on the reduced configuration space $\Theta_{K_{u,\gamma}}$ defined by the set \eqref{K-ug1} of configurational d.o.f.. It turns out \cite{q-nonl} that the measure $dx$ used to define the Hilbert space $\h_\lambda$ (see Section \ref{outline}) is just the coordinate measure $dx^I_{y_1}\ldots dx^J_{e_N}$. This means that it is rather easy to define operators $\{\hat{\kappa}^I_{y_1},\ldots,\hat{\kappa}^I_{y_M}\}$ on $\h_\lambda$ corresponding to $(\xi^I_{\pm s})$:
\[
\hat{\kappa}^I_{y_{\alpha}}\Psi:=x^I_{y_\alpha}\Psi,
\]    
where $\Psi\in\h_\lambda$. Consequently, it is also easy to define operators corresponding to $\xi^0_{\pm s}$ and $q_{IJ}$.  

On the other hand, we may try to keep the original form of the constraints by defining operators corresponding to $\theta^0,p_0$  and $\xi^0$ in terms of operators corresponding to $(\zeta_{\pm s I},r_J,\xi^K_{\pm s},\theta^L)$---taking into account the formulae \eqref{old-new} and \eqref{xi-0} with $\iota(\theta^I)=\pm \sgn(\theta^I)$ we see that the operators corresponding to  $\theta^0,p_0$  and $\xi^0$ can be defined only modulo the factor $\pm\sgn(\theta^I)$ but this is not an obstacle since in the original form of the constraints these variables appear always in pairs like e.g. $p_0\we\theta^0$ which means that the factor $\pm \sgn(\theta^I)$ is here irrelevant. 

Let us note finally that if it turned out that the variables $(\xi^I_\iota)$ could be gauge fixed to zero {\em globally} then this would amount to a considerable simplification of the constraints. Of course, while quantizing TEGR we would like to keep all degrees of freedom unfixed and use a gauge fixing like this only in the last resort.
  
\paragraph{Acknowledgments} This work was partially supported by the grant N N202 104838 of Polish Ministerstwo Nauki i Szkolnictwa Wy\.zszego.

\appendix

\section{Proof of Lemma \ref{TH-bij} \label{app-TH-bij}}

First let us state and prove an auxiliary lemma which will be used while proving Lemma \eqref{TH-bij}:
\begin{lm}
Let $(\theta^A)$ be a quadruplet of one-forms on $\cal M$ satisfying Condition 3b of Lemma \eqref{TH-bij}. Then for every function $(\phi^A)$ on $\cal M$ valued in $\mathbb{M}$ there exists a unique function $M$ on the manifold and a unique vector field $\vec{M}$ on $\cal M$ tangent to the foliation $\{\Sigma_t\}_{t\in\R}$ such that 
\begin{equation}
\phi^A=-M\frac{1}{3!}\eps^A{}_{BCD}*(\theta^B\we\theta^C\we\theta^D)+\vec{M}\lr\theta^A.
\label{phi-dec-0}
\end{equation}
\label{phi-dec-lm}
\end{lm}
\begin{proof}
Given a point $(t,y)\in\R\times\Sigma={\cal M}$, consider the following linear map
\begin{equation}
T_{(t,y)}\Sigma_t\ni Y\mapsto \theta^A(Y)\in\mathbb{M}.
\label{Y-thY}
\end{equation}
The map is {\em injective} and, equivalently, the image of this map is three-dimensional and, equivalently, the kernel of the map is zero-dimensional---otherwise there would exists a {\em non-zero} vector $Y\in T_{(t,y)}\Sigma_t$ such that
\[
q(Y,Y)=\eta_{AB}\theta^A(Y)\theta^B(Y)=0,
\] 
which would contradict Condition 3b. This fact implies that there exists exactly two vectors $\xi^A_{(t,y)}$ in $\mathbb{M}$ orthogonal to the image and normalized to $-1$.

Consequently, there exists exactly two continuous functions $(\xi^A):{\cal M}\to\mathbb{M}$ values of which are normalized vectors orthogonal to images of \eqref{Y-thY}. Clearly, for every $t$ each such a function satisfies \eqref{xi-df}.

Fix one of the two functions $(\xi^A)$. Properties of $(\xi^A)$ and the map \eqref{Y-thY} guarantee that every function $(\phi^A)$ can be {\em uniquely} decomposed into a function $M$ and a vector field $\vec{M}$ tangent to the foliation $\{\Sigma_t\}_{t\in\R}$:
\begin{equation}
\phi^A=M\xi^A+\theta^A(\vec{M})=M\xi^A+\vec{M}\lr \theta^A.
\label{phi-dec}
\end{equation}
On the other hand, one of the two (continuous) functions $(\xi^A)$ reads \cite{os}
 \begin{equation}
\xi^A=-\frac{1}{3!}\veps^A{}_{BCD}*(\theta^B\we\theta^C\we\theta^D).
\label{xi-*}
\end{equation}
\end{proof}

\begin{proof}[Proof of Lemma \ref{TH-bij}]
We will prove the lemma by showing that the map \eqref{NNth-bth} is a bijection. 

\paragraph{The map \eqref{NNth-bth} is injective} Assume that $(N,\vec{N},\theta^A)$ and $(N',\vec{N}',\theta^{\prime A})$ are mapped by \eqref{NNth-bth} to the same $(\bth^A)$. Then 
\[
\bed t[(N-N')\xi^A+\vec{N}\lr\theta^A-\vec{N}'\lr\theta^{\prime A}]+\theta^A-\theta^{\prime A}=0.
\]  
By virtue of Condition 3a and Lemma \eqref{phi-dec-lm} $N=N'$ and $\vec{N}=\vec{N}'$. This means that $\theta^A=\theta^{\prime A}$. 

\paragraph{The image of the map \eqref{NNth-bth} is contained in $\bld{\Theta}$} 
It was shown in \cite{os} that if $(\bth^A)$ is given by \eqref{NNth-bth} then:
\begin{equation}
\det(\bth^A_\alpha)= {N}\sqrt{\det (q_{ij})}.
\label{det-pm}
\end{equation}
Here $(\bth^A_\alpha)$, $\alpha=0,1,2,3$,  are components of $\bth^A$ in a (local) coordinate frame $(y^\alpha)\equiv (t,y^i)$, $i=1,2,3$,   on $\cal M$ compatible with its orientation defined by a (local) coordinate frame $(y^i)$ on $\Sigma$; $(q_{ij})$ are components of $q$ in the frame $(y^i)$ on $\Sigma_t$. By virtue of Condition 1 $\det(\bth^A_\alpha)>0$ which means that $(\bth^A)$ is a global coframe compatible with the orientation of $\cal M$. 

If $Y$ is a vector tangent to $\Sigma_t$ then $Y\lr\bed t=0$. Hence
\[
g(Y,Y)=\eta_{AB}(Y\lr\bth^A)(Y\lr\bth^A)=\eta_{AB}(Y\lr\theta^A)(Y\lr\theta^A)=q(Y,Y).
\]
By virtue of Condition 3b $\Sigma_t$ is spatial with respect to $g$.

\paragraph{The map \eqref{NNth-bth} is surjective} Let $(\bth^A)\in\bld{\Theta}$. $\bth^A$ is a sum \cite{ham-diff,mielke} of $\bed t \,\partial_t\lr\bth^A$ and a one-form $\partial_t\lr(\bed t\we\bth^A)=:\theta^A$ which satisfies Condition 3a. By virtue of Lemma \ref{phi-dec-lm} $(\partial_t\bth^A)$ define unambiguously a function $N$ and a vector field $\vec{N}$ which satisfies Condition 2. If $Y$ is a vector tangent to $\Sigma_t$ then
\[
\theta^A(Y)=Y\lr\theta^A=Y\lr\Big(\partial_t\lr(\bed t\we\bth^A)\Big)=\bth^A(Y)
\]  
because $Y\lr\bed t=0$. Hence, if $Y,Y'$ are tangent to $\Sigma_t$ at the same point then  
\[
{q}(Y,Y')=\eta_{AB}\theta^A(Y)\theta^B(Y')=\eta_{AB}\bth^A(Y)\bth^B(Y')=g(Y,Y').
\]
Recall that $\Sigma_t$ is space-like with respect to $g$. Therefore $(\theta^A)$ meets Condition 3b.

Obviously $(\bth^A)$ can be restored from the triplet $({N},\vec{N},\theta^A)$ by means of \eqref{NNth-bth}. By virtue of \eqref{det-pm} and Condition 3b the function $N$ is positive everywhere. 
\end{proof}

\section{Constraints of TEGR and YMTM in terms of new va\-riab\-les---derivation \label{constr-der}}

The goal of the present section is to rewrite the constraints of YMTM and TEGR presented in Section \ref{constr-old} in terms of the new variables $(\zeta_{\iota I},r_J,\xi_\iota^K,\theta^L)$---obviously, this also amounts to rewriting the Hamiltonians of the two theories. 

The Hamiltonian formulations of TEGR and YMTM in \cite{oko-tegr} and \cite{os} we derived under an assumption that $\Sigma$ is a compact manifold without boundary. Here we will keep the assumption.

As mentioned at the beginning of Section \ref{constr-new} we will use Equations \eqref{old-new} written in a bit simpler form 
\begin{equation}
\begin{aligned}
&p_0=\xi_\iota^0\vec{\theta}^I\lr\zeta_{\iota I},&& p_I=r_I-\xi_{\iota I}\,\vec{\theta}^J\lr\zeta_{\iota J},\\
&\theta^0=\frac{\xi_{\iota I}}{\xi_\iota^0}\,\theta^I, & & \theta^I=\theta^I.
\end{aligned} 
\label{old-new-1}
\end{equation}
Moreover, we will express $(\xi^A)$ defined by \eqref{xi-*} appearing in the constraints by means of the formulae \eqref{xi-0s} and \eqref{xi-Is}:
\begin{align}
\xi^0&=\sgn(\theta^J)|\xi^0_\iota|=\frac{\sgn(\theta^J)}{\iota(\theta^J)}\xi^0_\iota,& \xi^I&=\frac{\sgn(\theta^J)}{\iota(\theta^J)}\xi^I_\iota.
\label{xis-ximu}
\end{align} 
Let us emphasize that both in \eqref{old-new-1} and \eqref{xis-ximu} $\xi_\iota^0$ is not an independent variable but rather a function of $(\xi_\iota^I,\theta^J)$ given by \eqref{xi-0}. 

In the calculations below we will use also the following identity:
\begin{equation}
\vec{\theta}^I\lr\theta^J=\delta^{IJ}+\xi_\iota^I\xi_\iota^J=\bar{q}^{IJ},
\label{tt-q}
\end{equation}
where $(\bar{q}^{IJ})$ is the inverse matrix to $(q_{IJ})$ (see \eqref{bar-qIJ})---the first equality above follows immediately from the identity \cite{oko-tegr} 
\[
\vec{\theta}^A\lr\theta^B=\eta^{AB}+\xi^A\xi^B,
\]
where $(\xi^A)$ is any solution of \eqref{xi-df}. Let us emphasize that in all formulae below both $q_{IJ}$ and $\bar{q}^{IJ}$ will be consider as functions of $(\xi^I_\iota)$ (see \eqref{q-xi} and \eqref{bar-qIJ}).   

\subsection{The constraints of TEGR}

\subsubsection{The scalar constraint}

The scalar constraint \eqref{S(N)} of TEGR consists of five terms which will be transformed in turn.

\paragraph{The first term}  The first term can be rewritten as follows
\begin{multline*}
p_A\we\theta^B\we*(p_B\we\theta^A)=p_0\we\theta^0\we*(p_0\we\theta^0)+2p_I\we\theta^0\we*(p_0\we\theta^I)+p_I\we\theta^J\we*(p_J\we\theta^I)=\\=\vec{\theta}^I\lr\zeta_{\iota I}\we\xi_{\iota J}\theta^J\we*(\vec{\theta}^K\lr\zeta_{\iota K}\we\xi_{\iota L}\theta^L)+2\vec{\theta}^L\lr\zeta_{\iota L}\we\theta^I\we*(r_I\we\xi_{\iota K}\theta^K)-\\-2\vec{\theta}^L\lr\zeta_{\iota L}\we\theta^I\we*(\xi_{\iota I}\vec{\theta}^K\lr\zeta_{\iota K}\we\xi_{\iota J}\theta^J)+r_I\we\theta^J\we*(r_J\we\theta^I)-\\-2\xi_{\iota I}\vec{\theta}^K\lr\zeta_{\iota K}\we\theta^J\we*(r_J\we\theta^I)+\xi_{\iota I}\vec{\theta}^K\lr\zeta_{\iota K}\we\theta^J\we*(\xi_{\iota J}\vec{\theta}^L\lr\zeta_{\iota L}\we\theta^I)=\\=r_I\we\theta^J\we*(r_J\we\theta^I).
\end{multline*}
\paragraph{The second term} To express the second term in \eqref{S(N)} as a function of the new variables it is enough to transform the factor $p_A\we\theta^A$:
\begin{equation*}
p_A\we\theta^A=p_0\we\theta^0+p_I\we\theta^I=\vec{\theta}^I\lr\zeta_{\iota I}\we\xi_{\iota J}\theta^J+(r_I-\xi_{\iota I}\vec{\theta}^K\lr\zeta_{\iota K})\we\theta^I=r_I\we\theta^I.
\end{equation*}

\paragraph{The third term} Consider the integrated third term: 
\begin{equation}
-\int_\Sigma M\xi^Adp_A=\int_\Sigma d(M\xi^A)\we p_A=\int_\Sigma dM\we \xi^Ap_A+M d\xi^A\we p_A.
\label{Nxip}
\end{equation}
Let us focus on the first of the two resulting terms:
\begin{multline}
\xi^Ap_A=\frac{\sgn(\theta^J)}{\iota(\theta^J)}(\xi_\iota^0p_0+\xi_\iota^Ip_I)=\frac{\sgn(\theta^J)}{\iota(\theta^J)}\big((\xi_\iota^0)^2\vec{\theta}^J\lr\zeta_{\iota J}+\xi_\iota^I(r_I-\xi_{\iota I}\vec{\theta}^J\lr \zeta_{\iota J})\big)=\\=\frac{\sgn(\theta^J)}{\iota(\theta^J)}\Big([(\xi_\iota^0)^2-\xi_\iota^I\xi_{\iota I}]\vec{\theta}^J\lr\zeta_{\iota J}+\xi_\iota^Ir_I\Big)=\frac{\sgn(\theta^J)}{\iota(\theta^J)}\big(-\xi_\iota^A\xi_{\iota A}\vec{\theta}^J\lr\zeta_{\iota J}+\xi_\iota^Ir_I\big)=\\=\frac{\sgn(\theta^J)}{\iota(\theta^J)}(\vec{\theta}^J\lr\zeta_{\iota J}+\xi_\iota^Ir_I),
\label{xi-p}
\end{multline}
where in the last step we used the second equation in \eqref{xi-df}. The last term in \eqref{Nxip}
\begin{multline*}
d\xi^A\we p_A=\frac{\sgn(\theta^J)}{\iota(\theta^J)}(d\xi_\iota^0\we p_0+d\xi_\iota^I\we p_I)=\frac{\sgn(\theta^J)}{\iota(\theta^J)}\big(\xi_\iota^0d\xi_\iota^0\we\vec{\theta}^J\lr\zeta_{\iota J}+\\+d\xi_\iota^I\we(r_I-\xi_{\iota I}\vec{\theta}^J\lr \zeta_{\iota J})\big)=\frac{\sgn(\theta^J)}{\iota(\theta^J)}\Big([\xi_\iota^0d\xi_\iota^0-\xi_{\iota I}d\xi_\iota^I]\we\vec{\theta}^J\lr\zeta_{\iota J}+d\xi_\iota^I\we r_I\Big)=\\=\frac{\sgn(\theta^J)}{\iota(\theta^J)}d\xi_\iota^I\we r_I
\end{multline*}
---here in the last step we applied the identity
\[
\xi_\iota^0d\xi_\iota^0-\xi_{\iota I}d\xi_\iota^I=-\xi_\iota^Ad\xi_{\iota A}=0
\]
which follows from the second equation in \eqref{xi-df}. Thus
\begin{multline}
-\int_\Sigma M\xi^Adp_A=\frac{\sgn(\theta^J)}{\iota(\theta^J)}\int_\Sigma dM\we (\vec{\theta}^J\lr\zeta_{\iota J}+\xi_\iota^Ir_I)+Md\xi_\iota^I\we r_I=\\=\frac{\sgn(\theta^J)}{\iota(\theta^J)}\int_\Sigma dM\we \vec{\theta}^J\lr\zeta_{\iota J}+d(M\xi_\iota^I)\we r_I=-\frac{\sgn(\theta^J)}{\iota(\theta^J)}\int_\Sigma M\Big( d(\vec{\theta}^J\lr\zeta_{\iota J})+\xi_\iota^I\we dr_I\Big). 
\label{Nxip-fin}
\end{multline}

\paragraph{The fourth term}  The next term in \eqref{S(N)} can be expressed as follows:
\begin{multline*}
d\theta_A\we\theta^B\we*(d\theta_B\we\theta^A)=d\theta^0\we\theta^0\we*(d\theta^0\we\theta^0)-2d\theta^0\we\theta^I\we*(d\theta_I\we\theta^0)+\\+d\theta_I\we\theta^J\we*(d\theta_I\we\theta^J).
\end{multline*} 
Let us express first $d\theta^0$ as a function of $(\xi_\iota^I,\theta^J)$ and their exterior derivatives:
\begin{multline}
d\theta^0=d\Big(\frac{\xi_{\iota I}}{\xi_\iota^0}\theta^I\Big)=d(\alpha_I\theta^I)=d\alpha_I\we\theta^I+\alpha_Id\theta^I=\\=\frac{\partial\alpha_I}{\partial\xi_\iota^J}d\xi_\iota^J\we\theta^I+\frac{\xi_{\iota I}}{\xi_\iota^0}d\theta^I=\frac{1}{\xi_\iota^0}\big(q_{IJ}d\xi_\iota^J\we\theta^I+{\xi_{\iota I}}d\theta^I\big)
\label{dt0}
\end{multline}
---in these calculations we used the second formula in \eqref{xi-al} and \eqref{al/xi}. Now let us calculate the factor $d\theta^0\we\theta^0$---by virtue of \eqref{dt0} and \eqref{q-xi}
\begin{multline}
d\theta^0\we\theta^0=({\xi_\iota^0})^{-2}({q_{IJ}}d\xi_\iota^J\we\theta^I+{\xi_{\iota I}}d\theta^I)\we\xi_{\iota K}\theta^K=\\=({\xi_\iota^0})^{-2}\Big([\delta_{IJ}-(\xi_\iota^0)^{-2}\xi_{\iota I}\xi_{\iota J}]d\xi_\iota^J\we\theta^I+{\xi_{\iota I}}d\theta^I\Big)\we\xi_{\iota K}\theta^K=\\=({\xi_\iota^0})^{-2} (d\xi_{\iota I}\we\theta^I+{\xi_{\iota I}}d\theta^I)\we\xi_{\iota K}\theta^K=({\xi_\iota^0})^{-2} d(\xi_{\iota I}\theta^I)\we\xi_{\iota K}\theta^K
\label{dt0-t0}
\end{multline}
since $\xi_{\iota I}\theta^I\we\xi_{\iota K}\theta^K=0$. Thus 
\begin{multline*}
d\theta_A\we\theta^B\we*(d\theta_B\we\theta^A)=({\xi_\iota^0})^{-4} d(\xi_{\iota I}\theta^I)\we\xi_{\iota J}\theta^J\we*(d(\xi_{\iota K}\theta^K)\we\xi_{\iota L}\theta^L)-\\-2(\xi_\iota^0)^{-2}(q_{IJ}d\xi_\iota^J\we\theta^I+\xi_{\iota I} d\theta^I)\we\theta^K\we*(d\theta_K\we\xi_{\iota L}\theta^L)+d\theta_I\we\theta^J\we*(d\theta_I\we\theta^J).
\end{multline*}

\paragraph{The fifth term} Regarding the fifth term in \eqref{S(N)} it is enough to calculate  
\begin{equation*}
d\theta_A\we\theta^A=-d\theta^0\we\theta^0+d\theta_I\we\theta^I=-({\xi_\iota^0})^{-2} d(\xi_{\iota I}\theta^I)\we\xi_{\iota K}\theta^K+d\theta_I\we\theta^I,
\end{equation*}
where we used \eqref{dt0-t0}.

Gathering all the partial results we obtain \eqref{scal-c}.

\subsubsection{The vector constraint}

Let us now turn to the vector constraint \eqref{V(N)}. It was shown in \cite{os} that
\[
V(\vec{M})=\int_\Sigma\LL_{\vec{M}}\theta^A\we p_A=\int_\Sigma\LL_{\vec{M}}\theta^0\we p_0+\LL_{\vec{M}}\theta^I\we p_I,
\]
where $\LL_{\vec{M}}$ denotes the Lie derivative on $\Sigma$ with respect to the vector field $\vec{M}$. It is easy to check that
\[
\LL_{\vec{M}}\theta^0=\frac{q_{IJ}}{\xi_\iota^0}(\LL_{\vec{M}}\xi_\iota^J)\theta^I+\frac{\xi_{\iota I}}{\xi_\iota^0}\LL_{\vec{M}}\theta^I
\]    
(see \eqref{dot-t} and \eqref{dt0}). Thus
\begin{multline*}
V(\vec{M})=\int_\Sigma\Big(\frac{q_{IJ}}{\xi_\iota^0}(\LL_{\vec{M}}\xi_\iota^J)\theta^I+\frac{\xi_{\iota I}}{\xi_\iota^0}\LL_{\vec{M}}\theta^I\Big)\we \xi_\iota^0\vec{\theta}^K\lr\zeta_{\iota K}+\LL_{\vec{M}}\theta^I\we(r_I-\xi_{\iota I}\,\vec{\theta}^J\lr\zeta_{\iota J})=\\=\int_\Sigma {q_{IJ}}(\LL_{\vec{M}}\xi_\iota^J)\theta^I\we \vec{\theta}^K\lr\zeta_{\iota K}+\LL_{\vec{M}}\theta^I\we r_I=\int_\Sigma {q_{IJ}}(\LL_{\vec{M}}\xi_\iota^J)\vec{\theta}^K\lr\theta^I \we \zeta_{\iota K}+\LL_{\vec{M}}\theta^I\we r_I=\\=\int_\Sigma (\LL_{\vec{M}}\xi_\iota^K)\zeta_{\iota K}+\LL_{\vec{M}}\theta^I\we r_I,
\end{multline*}
where in the last step we used \eqref{tt-q}. The well known expression 
\[
\LL_{\vec{M}}=d\circ\vec{M}\lr+\vec{M}\lr \circ d
\] 
allows us to rewrite the result above in a form \eqref{v(n)-new} free of derivatives of $\vec{M}$ (see also \cite{os}).

\subsubsection{The boost constraint}   

Using \eqref{dt0}, \eqref{xi-p} and \eqref{q-xi} we obtain
\begin{multline*}
B(a)=\int_\Sigma a\we(\theta^A\we*d\theta_A+\xi^Ap_A)=\int_\Sigma a\we(-\theta^0\we*d\theta^0+\theta^I\we*d\theta_I+\xi^Ap_A)=\\=\int_\Sigma a\we\Big(-\frac{1}{(\xi_\iota^0)^2}\xi_{\iota I}\theta^I\we*(q_{JK}d\xi_\iota^J\we\theta^K+\xi_{\iota J}d\theta^J)+\theta^I\we*d\theta_I+\frac{\sgn(\theta^L)}{\iota(\theta^L)}\big(\vec{\theta}^I\lr\zeta_{\iota I}+\xi_\iota^Ir_I\big)\Big)=\\=\int_\Sigma a\we\Big(-\frac{\xi_{\iota I}q_{JK}}{(\xi_\iota^0)^2}\theta^I\we*(d\xi_\iota^J\we\theta^K)+q_{IJ}\theta^I\we*d\theta^J+\frac{\sgn(\theta^L)}{\iota(\theta^L)}\big(\vec{\theta}^I\lr\zeta_{\iota I}+\xi_\iota^I r_I\big)\Big),
\end{multline*}
which coincides with \eqref{boost-c}.

\subsubsection{The rotation constraint}

By virtue of \eqref{dt0} 
\begin{multline*}
R(b)=\int_\Sigma b\we(\theta^A\we*p_A-\xi^Ad\theta_A)=\int_\Sigma b\we \Big( \xi_{\iota I}\theta^I\we*(\vec{\theta}^J\lr\zeta_{\iota J})+\theta^I\we*(r_I-\xi_{\iota I}\vec{\theta}^J\lr\zeta_{\iota J})+\\+\frac{\sgn(\theta^L)}{\iota(\theta^L)}\big(q_{IJ}d\xi_\iota^I\we\theta^J+\xi_{\iota I}d\theta^I-\xi_\iota^I d\theta_I\big)\Big)=\int_\Sigma b\we(\theta^I\we*r_I+\frac{\sgn(\theta^L)}{\iota(\theta^L)}q_{IJ}d\xi_\iota^I\we\theta^J),
\end{multline*}
which coincides with \eqref{rot-c}.

\subsection{The constraints of YMTM}

It is enough to rewrite the scalar constraint \eqref{s(N)} since the vector constraint $v(\vec{M})$ of YMTM coincides with the vector constraint $V(\vec{M})$ of TEGR. The scalar constraint consists of three terms which will be transformed in turn.

\paragraph{The first term} The first term \eqref{s(N)} can be rewritten as follows:
\begin{multline}
p^A\we* p_A=-p_0\we* p_0+p^I\we* p_I=\\=-(\xi_\iota^0)^2\vec{\theta}^I\lr\zeta_{\iota I}\we*(\vec{\theta}^J\lr\zeta_{\iota J})+(r^I-\xi_\iota^I\vec{\theta}^K\lr\zeta_{\iota K})\we*(r_I-\xi_{\iota I}\vec{\theta}^L\lr\zeta_{\iota L})=\\=-(\xi_\iota^0)^2\vec{\theta}^I\lr\zeta_{\iota I}\we*(\vec{\theta}^J\lr\zeta_{\iota J})+r^I\we*r_I-2r^I\we*(\xi_{\iota I}\vec{\theta}^K\lr\zeta_{\iota K})+\\+\xi_\iota^I\xi_{\iota I}\,\vec{\theta}^K\lr\zeta_{\iota K}\we*(\vec{\theta}^L\lr\zeta_{\iota L})=-\vec{\theta}^I\lr\zeta_{\iota I}\we*(\vec{\theta}^J\lr\zeta_{\iota J})+r^I\we*r_I-2\xi_{\iota I}r^I\we*(\vec{\theta}^K\lr\zeta_{\iota K}),
\label{pp-0}
\end{multline}
where in the last step we applied the second equation in \eqref{xi-df}. Using \eqref{a-b} and the fact that $*\zeta_{\iota I}$ is a zero-form, that is, a function we can transform the first term of the result above as follows:
\begin{multline}
-\vec{\theta}^I\lr\zeta_{\iota I}\we*(\vec{\theta}^J\lr\zeta_{\iota J})=-*(*\zeta_{\iota I}\we\theta^I)\we *\zeta_{\iota J}\we\theta^J= -*\zeta_{\iota I}*\zeta_{\iota J}*\theta^I\we\theta^J=\\=-*\zeta_{\iota I}*\zeta_{\iota J}**(*\theta^I\we\theta^J)=-\zeta_{\iota I}*\zeta_{\iota J}\,\vec{\theta}^J\lr\theta^I=-\bar{q}^{IJ}\zeta_{\iota I}*\zeta_{\iota J},
\label{tz*tz}
\end{multline}
where in the last step we used \eqref{tt-q}. Let us also simplify the last term in \eqref{pp-0}: again by virtue of \eqref{a-b}
\begin{equation*}
-2\xi_{\iota I}r^I\we*(\vec{\theta}^K\lr\zeta_{\iota K})=-2\xi_{\iota I}r^I\we *\zeta_{\iota K} \theta^K=-2\xi_{\iota I}*\zeta_{\iota K} r^I\we{\theta}^K. 
\end{equation*}
Setting this result and \eqref{tz*tz} to \eqref{pp-0} gives 
\begin{equation}
p^A\we* p_A=-\bar{q}^{IJ}\zeta_{\iota I}*\zeta_{\iota J}+r^I\we*r_I-2\xi_\iota^I*\zeta_{\iota K}\, r_I\we{\theta}^K.
\label{pp}
\end{equation}

\paragraph{The second term} The second term in \eqref{s(N)} is already calculated---see \eqref{Nxip-fin}.

\paragraph{The third term} To calculate the third term in \eqref{s(N)} in terms of the new variables we apply \eqref{dt0}:
\begin{multline*}
d\theta^A\we*d\theta_A=-d\theta^0\we*d\theta^0+d\theta^I\we*d\theta_I=-\frac{q_{IJ}q_{KL}}{(\xi_\iota^0)^2}d\xi_\iota^I\we\theta^J\we*(d\xi_\iota^K\we\theta^L)-\\-\frac{2q_{IJ}}{(\xi_\iota^0)^2}d\xi_\iota^I\we\theta^J\we*(\xi_{\iota K}d\theta^K)-\frac{\xi_{\iota I}\xi_{\iota J}}{(\xi_\iota^0)^2}d\theta^I\we*d\theta^J+d\theta^I\we*d\theta_I=\\=-\frac{q_{IJ}q_{KL}}{(\xi_\iota^0)^2}d\xi_\iota^I\we\theta^J\we*(d\xi_\iota^K\we\theta^L)-\frac{2q_{IJ}}{(\xi_\iota^0)^2}d\xi_\iota^I\we\theta^J\we*(\xi_{\iota K}d\theta^K)+q_{IJ}d\theta^I\we*d\theta^J,
\end{multline*} 
where in the last step we applied \eqref{xi-0} and \eqref{q-xi}.

Gathering the partial results we obtain \eqref{scal-c-ymtm}.




\begin{thebibliography}{}

%
\bibitem{q-suit} Oko{\l}\'ow A 2013 Variables suitable for constructing quantum states for the Teleparallel Equivalent of General Relativity I {\it Gen. Rel. Grav.} {\bf 46} 1620 {\it E-print} \verb+arXiv:1305.4526+
%
\bibitem{q-nonl} Oko{\l}\'ow A 2013 Construction of spaces of kinematic quantum states for field theories via projective techniques {\it Class. Quant. Grav.} {\bf 30} 195003 {\it E-print} \verb+arXiv:1304.6330+ 
%
\bibitem{oko-tegr} Oko{\l}\'ow A 2013 ADM-like Hamiltonian formulation of gravity in the teleparallel geometry {\it Gen. Rel. Grav.} {\bf 45} 2569--2610 {\it E-print} \verb+arXiv:1111.5498+ 
%
\bibitem{bl} Blagojevi\'c M, Nikoli\'c I A 2000 Hamiltonian structure of the teleparallel formulation of GR {\it Phys. Rev. D} {\bf 62} 024021 {\it E-print} \verb+arXiv:hep-th/0002022+
%
\bibitem{maluf-1} Maluf J W, da Rocha-Neto J F  2001 Hamiltonian formulation of general relativity in the teleparallel geometry {\it Phys. Rev. D} {\bf 64} 084014 {\it E-print} \verb+arXiv:gr-qc/0002059+
%
\bibitem{maluf} da Rocha-Neto J F, Maluf J W and Ulhoa S C 2010 Hamiltonian formulation of unimodular gravity in the teleparallel geometry {\it Phys. Rev. D} {\bf 82} 124035 {\it E-print} \verb+arXiv:1101.2425+
%
\bibitem{q-stat} Oko{\l}\'ow A 2013 Kinematic quantum states for the Teleparallel Equivalent of General Relativity {\it Gen. Rel. Grav.} {\bf 46} 1653 {\it E-print} \verb+arXiv:1304.6492+
%
\bibitem{itin} Itin Y 2002 Conserved currents for general teleparallel models {\it Int. J. Mod. Phys.} {\bf 17} 2765 {\it E-print} \verb+arXiv:gr-qc/0103017+
%
\bibitem{os} Oko{\l}\'ow A, \'Swie\.zewski J 2012 Hamiltonian formulation of a simple theory of the teleparallel geometry  {\it Class. Quant. Grav.}  {\bf 29}  045008   {\it Preprint}  \verb+arXiv:1111.5490+ 
%
\bibitem{nester} Nester J M 1989 Positive energy via the teleparallel Hamiltonian {\it Int. J. Mod. Phys. A} {\bf 4} 1755-1772
%
\bibitem{ham-diff} Wallner R P 1990 New variables in gravity theories {\it Phys. Rev. D} {\bf 42} 441-448
%
\bibitem{mielke} Mielke E W 1992 Ashtekar's Complex Variables in General Relativity and Its Teleparallelism Equivalent {\it Ann. Phys.} {\bf 219} 78-108
%
\bibitem{mal-rev} Maluf J W 2013 The teleparallel equivalent of general relativity {\it E-print} \verb+arXiv:+ \verb+1303.3897+
%
\bibitem{adm} Arnowitt R, Deser S, Misner C W 1962 The Dynamics of General Relativity  {\em Gravitation: an introduction to current research}, Louis Witten ed. (Wiley 1962), chapter 7, 227--265 {\it E-print} \verb+arXiv:gr-qc/0405109+
%

\end{thebibliography}
\end{document}